%% file: main.tex
\documentclass[journal,10pt,twocolumn,oneside]{IEEEtran}
\usepackage{times}
\usepackage{epsfig}
\usepackage{graphicx}
\usepackage{amsmath}
\allowdisplaybreaks
\usepackage{amssymb}
\usepackage{algorithm}
\usepackage{algorithmic}
\usepackage{enumerate}
\usepackage{multirow}
\usepackage{multicol}
\usepackage{arydshln}
\usepackage{color}
\usepackage{amssymb}
\usepackage{amsthm}
\input{defs}

\newtheorem{theorem}{Theorem}

\newtheorem{lemma}{Lemma}

\newtheorem{corollary}{Corollary}

\newcommand{\red}[1] {\textcolor[rgb]{0.0,0.0,0.0}{{#1}}}  

\newcommand{\blue}[1] {\textcolor[rgb]{0.0,0.0,0.0}{{#1}}}  

\newcommand\arXiv{1}

\begin{document}

\title{Manifold Graph Signal Restoration using \\
Gradient Graph Laplacian Regularizer}
%
\author{
\IEEEauthorblockN{Fei Chen, \emph{Member, IEEE}, Gene Cheung, \emph{Fellow, IEEE}, Xue Zhang, \emph{Member, IEEE}}
\renewcommand{\baselinestretch}{1.0}
\thanks{The work of F. Chen was supported in part by the National Natural Science Foundation of China (61771141) and the Natural Science Foundation of Fujian Province (2021J01620). The work of G. Cheung was supported in part by the Natural Sciences and Engineering Research Council of Canada (NSERC) RGPIN-2019-06271, RGPAS-2019-00110. The work of X. Zhang was supported in part by the National Natural Science Foundation of China (62302278), the Natural Science Foundation of Shandong Province (ZR2023QF014) and the Natural Science Foundation of Qingdao Municipality (23-2-1-112-zyyd-jch). \textit{(Corresponding author: Xue Zhang)}}
\thanks{F. Chen is with College of Computer and Data Science, Fuzhou University, Fuzhou, China (e-mail:chenfei314@fzu.edu.cn).}
\thanks{ G. Cheung is with the department of EECS, York University, 4700 Keele Street, Toronto, M3J 1P3, Canada (e-mail: genec@yorku.ca).} 
\thanks{X. Zhang is with the College of Computer Science and Engineering, Shandong University of Science and Technology, Qingdao 266590, China (e-mail: xuezhang@sdust.edu.cn).}
}
%
%
%

\maketitle
\begin{abstract}
\input{abstract}
\end{abstract}
\section{Introduction}
\label{sec:intro}
\input{intro.tex}

\section{Related Works}
\label{sec:related}
\input{related.tex}

\section{Preliminaries}
\label{sec:prelim}
\input{prelim.tex}

\section{GGLR for Graph with Coordinates}
\label{sec:gglr}
\input{GGLR2.tex}

\section{Computation of Optimal Tradeoff Parameter}
\label{sec:optPara}
\input{optPara.tex}

\section{GGLR for Graphs without Coordinates}
\label{sec:general}
\input{general.tex}

\section{Experiments}
\label{sec:results}
\input{results}

\section{Conclusion}
\label{sec:conclude}
\input{conclude}

\appendix
\input{append}
\bibliographystyle{IEEEbib}
\bibliography{ref2}

\ifnum\arXiv=0
    \input{bio.tex} 
\fi

\end{document}

%% file: defs.tex
\def\0{{\mathbf 0}}
\def\1{{\mathbf 1}}

\def\d{{\mathbf d}}
\def\e{{\mathbf e}}
\def\f{{\mathbf f}}

\def\h{{\mathbf h}}

\def\p{{\mathbf p}}
\def\q{{\mathbf q}}

\def\s{{\mathbf s}}

\def\u{{\mathbf u}}
\def\v{{\mathbf v}}

\def\x{{\mathbf x}}
\def\y{{\mathbf y}}
\def\z{{\mathbf z}}

\def\A{{\mathbf A}}

\def\C{{\mathbf C}}
\def\D{{\mathbf D}}

\def\F{{\mathbf F}}
\def\G{{\mathbf G}}
\def\F{{\mathbf F}}
\def\H{{\mathbf H}}
\def\I{{\mathbf I}}

\def\L{{\mathbf L}}
\def\M{{\mathbf M}}

\def\P{{\mathbf P}}
\def\Q{{\mathbf Q}}
\def\R{{\mathbf R}}

\def\U{{\mathbf U}}
\def\V{{\mathbf V}}
\def\W{{\mathbf W}}

\def\ie{{\textit{i.e.}}}
\def\eg{{\textit{e.g.}}}

\def\cE{{\mathcal E}}

\def\cG{{\mathcal G}}
\def\cH{{\mathcal H}}

\def\cL{{\mathcal L}}

\def\cN{{\mathcal N}}
\def\cO{{\mathcal O}}
\def\cP{{\mathcal P}}

\def\cS{{\mathcal S}}
\def\cT{{\mathcal T}}

\def\cV{{\mathcal V}}

\def\balpha{{\boldsymbol \alpha}}

\def\bPhi{{\boldsymbol \Phi}}

\def\bSigma{{\boldsymbol \Sigma}}
\def\bTheta{{\boldsymbol \Theta}}

%% file: abstract.tex
In the graph signal processing (GSP) literature, graph Laplacian regularizer (GLR) was used for signal restoration to promote piecewise smooth / constant reconstruction with respect to an underlying graph. 
However, for signals slowly varying across graph kernels, GLR suffers from an undesirable ``staircase" effect. 
In this paper, focusing on manifold graphs---collections of uniform discrete samples on low-dimensional continuous manifolds---we generalize GLR to gradient graph Laplacian regularizer (GGLR) that promotes planar / piecewise planar (PWP) signal reconstruction.
Specifically, for a graph endowed with sampling coordinates (e.g., 2D images, 3D point clouds), we first define a gradient operator, using which we construct a gradient graph for nodes' gradients in the sampling manifold space.
This maps to a gradient-induced nodal graph (GNG) and a positive semi-definite (PSD) Laplacian matrix with planar signals as the $0$ frequencies. 
For manifold graphs without explicit sampling coordinates, we propose a graph embedding method to obtain node coordinates via fast eigenvector computation.
We derive the means-square-error minimizing weight parameter for GGLR efficiently, trading off bias and variance of the signal estimate.
Experimental results show that GGLR outperformed previous graph signal priors like GLR and graph total variation (GTV) in a range of graph signal restoration tasks.

%% file: intro.tex
Due to cost, complexity and limitations of signal sensing, acquired signals are often imperfect with distortions and/or missing samples. 
\textit{Signal restoration} is the task of recovering a pristine signal from corrupted and/or partial observations.
We focus on restoration of \textit{graph signals}---sets of discrete samples residing on graph-structured data kernels---studied in the \textit{graph signal processing} (GSP) field \cite{ortega18ieee,cheung21}. 
Examples of graph signal restoration include denoising \cite{pang17,zeng20,dinesh20}, dequantization \cite{liu17,liu19}, deblurring \cite{bai19}, and interpolation \cite{chen21}.

As an under-determined problem, signal restoration requires appropriate signal priors for regularization. 
Given a graph kernel encoded with pairwise similarities or correlations as edge weights, there exist numerous \textit{graph smoothness priors} that assume a target signal $\x \in \mathbb{R}^N$ is ``smooth" with respect to (w.r.t.) to the underlying graph $\cG$ in various mathematical forms\footnote{A review of graph smoothness priors is presented in Section\;\ref{subsec:smoothness}.} \cite{elmoataz08,couprie13,chen15,pang17,bai19,zeng20,dinesh20,chen21}.
Among them, the most common is the \textit{graph Laplacian regularizer} (GLR) \cite{pang17}, which minimizes $\x^\top \L \x$ assuming signal $\x$ is smooth w.r.t. to a graph $\cG$ specified by \textit{positive semi-definite} (PSD) graph Laplacian matrix $\L \in \mathbb{R}^{N \times N}$. 
GLR is popular partly because of its \textit{spectral interpretation} \cite{pang17}: minimizing $\x^\top \L \x$ means ``promoting"\footnote{\blue{By ``promote" we mean adding a non-negative signal prior $\text{Pr}(\x)$ to a minimization objective $f(\x)$ to \textit{bias} the solution towards $\x$ where $\text{Pr}(\x)=0$. See \cite{Belkin2018ReconcilingMM} for a comprehensive overview on bias-variance tradeoffs.}} (biasing the solution towards)  low graph frequencies---eigenvectors of $\L$ corresponding to small non-negative eigenvalues of $\L$, where the $0$ frequency for a positive connected graph $\cG$ without self-loops is the constant signal. 
It is also popular because of its ease in optimization---when combined with an $\ell_2$-norm fidelity term, it results in a system of linear equations for the optimal solution $\x^*$ that can be computed efficiently using fast numerical methods like \textit{conjugate gradient} (CG) \cite{axelsson1986rate} without matrix inverse. 

Further, a \textit{signal-dependent} (SD) variant of GLR---where each graph edge weight $w_{i,j}$ is inversely proportional to the difference in sought signal samples, $|x_i - x_j|$, similarly done in the \textit{bilateral filter} \cite{tomasi98}---has been shown to promote  
\textit{piecewise constant} (PWC) signal reconstruction \cite{pang17,liu17}. 
SDGLR has been used successfully in a number of practical applications: image denoising \cite{pang17}, JPEG image dequantization \cite{liu17}, 3D point cloud denoising \cite{dinesh20}, etc. 

\begin{figure}
\centering
\begin{tabular}{c}
\includegraphics[width=3.2in]{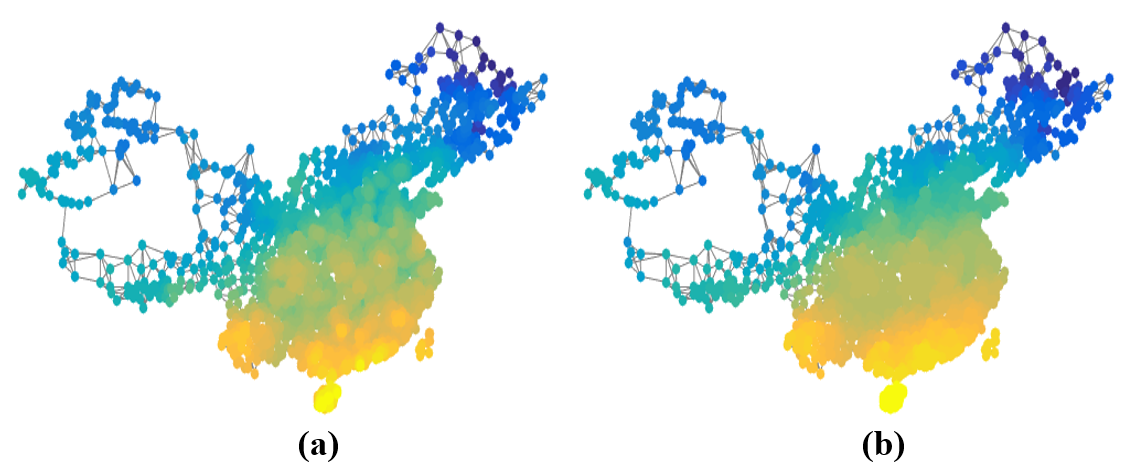}
\vspace{-0.1in}
\end{tabular}
\caption{\small Examples of denoising smooth graph signals: 
a) result (PSNR:$27.13$dB) using SDGLR on temperature signal on China graph with noise variance $50$; b) result (PSNR:$28.87$dB) using SDGGLR. ``Staircase" effect of SDGLR---artificial boundaries between adjacent constant pieces---compared to SDGGLR is visually obvious.}
\label{fig:PWP_ex}
\end{figure}

However, like \textit{total variation} (TV) \cite{chambolle1997image} common in image denoising that also promotes PWC signal reconstruction (but at a higher computation cost due to its non-differentiable $\ell_1$-norm \cite{chen21}), GLR suffers from the well-known adverse ``staircase" effect for signals with gradually changing intensity---artificially created boundaries between adjacent constant pieces that do not exist in the original signals. 
Fig.\;\ref{fig:PWP_ex}(a) shows a denoising example of a noise-corrupted temperature signal on a geographical graph of China, where the resulting blocking artifacts between adjacent reconstructed constant pieces are visually obvious. 
For images on 2D grids, as one remedy TV was extended to \textit{total generalized variation} (TGV) \cite{bredies2010total, bredies2015tgv} to promote piecewise linear / planar (PWP) signal reconstruction---a signal with piecewise constant gradients.

Surprisingly, to our knowledge no such generalized graph signal prior exists in GSP to promote PWP graph signals\footnote{\blue{While TV promoting PWC signals has been adapted in graph domains as \textit{graph total variation} (GTV) \cite{elmoataz08,couprie13}, there is no graph-equivalent of TGV.}}.
We believe that the difficulty in defining planar / PWP signals on graphs lies in the reliance on the conventional notion of \textit{graph gradient} computed using the Laplacian operator (to be dissected in Section\;\ref{subsec:whyNot}).
Instead, \textbf{we define planar / PWP graph signals using \textit{manifold gradients} on \textit{manifold graphs}}, where each graph node is a discrete sample on a continuous manifold with a sampling coordinate \cite{costa04,pang17}---node coordinates and graph structure are intimately related via \textit{graph embeddings}\footnote{A graph embedding computes node coordinates so that Euclidean distances between coordinate pairs reflect node-to-node distance in the original graph.} \cite{xu21}.
Manifold gradients are computed at graph nodes using sampling coordinates in a neighborhood, and a prior promoting constant / PWC manifold gradient reconstruction can lead to planar / PWP signal reconstruction.

Specifically, extending GLR, in this paper we propose a higher-order graph smoothness prior called \textit{gradient graph Laplacian regularizer} (GGLR) that promotes PWP signal reconstruction for  restoration tasks\footnote{An earlier conference version of this work focused exclusively on image interpolation \cite{chen21}.}.  
We divide our investigation into two parts.
In the first part, we assume that each graph node $i$ is endowed with a coordinate vector $\p_i \in \mathbb{R}^K$ in a manifold space of dimension $K \ll N$.
Examples include 2D images and 3D point clouds.
Using these node coordinates, for a node $i$ we define a \textit{gradient operator} $\F^i$, using which we compute manifold gradient $\balpha^i$ at node $i$ w.r.t. nodes in a local neighborhood.
To connect $\{\balpha^i\}$, we construct a positive \textit{gradient graph} $\bar{\cG}$ with Laplacian $\bar{\L}^o$, from which we define GGLR $\balpha^\top \bar{\L}^o \balpha = \x^\top \cL \x$, where $\cL$ is the \textit{gradient-induced nodal graph} (GNG) Laplacian.
We prove GGLR's promotion of planar signal reconstruction---specifically, \textbf{planar signals are the $0$ frequencies of PSD matrix $\cL$}.
We then show that its signal-dependent variant SDGGLR promotes PWP signal reconstruction.


In the second part, we assume the more general case where nodes in the manifold graph do not have accompanied sampling coordinates. 
We propose a \textit{parameter-free} graph embedding method \cite{Chen2022DSLW}, based on fast eigenvector computation \cite{lobpcg}, to compute manifold coordinates $\p_i$ for nodes $i$, using which we can define GGLR for regularization in a signal restoration task, as done in part one.

Crucial in regularization-based optimization is the selection of the tradeoff parameter $\mu$ that weighs the regularizer against the fidelity term.
Extending the analysis in \cite{chen2017bias}, we derive the means-square-error (MSE) minimizing $\mu$ for GGLR efficiently by trading off bias and variance of the signal estimate. 

We tested GGLR in four different practical applications: image interpolation, 3D point cloud color denoising, age estimation, and player rating estimation. 
Experimental results show that GGLR outperformed previous graph signal priors like GLR \cite{pang17} and graph total variation (GTV) \cite{elmoataz08,couprie13}.

We summarize our technical contributions as follows:
\begin{enumerate}
\item For graph nodes endowed with sampling coordinates, we define the higher-order Gradient Graph Laplacian Regularizer (GGLR) for graph signal restoration, which retains its quadratic form and thus ease of computation under specified conditions (Theorem\;\ref{thm:closedForm}).

\item We prove GGLR's promotion of planar signal reconstruction (Theorem\;\ref{thm:GGLR})---planar signals are $0$ frequencies of PSD GNG Laplacian $\cL$. 
We demonstrate that its signal-dependent variant SDGGLR promotes PWP signal reconstruction.


\item For manifold graphs without sampling coordinates, we propose a fast graph embedding method (Algorithm\;\ref{alg:alg_embedding}) to first obtain node coordinates via fast eigenvector computation before employing SDGGLR.

\item We derive the MSE-minimizing weight parameter $\mu$ for GGLR in a MAP formulation, trading off bias and variance of the signal estimate (Theorem\;\ref{thm:MSE}).

\item We demonstrate the efficacy of GGLR in four real-world graph signal restoration applications, outperforming previous graph smoothness priors. 
\end{enumerate}

The outline of the paper is as follows.
We overview related works in Section\;\ref{sec:related} and review basic definitions in Section\;\ref{sec:prelim}. 
We define GGLR in Section\;\ref{sec:gglr}.
We derive an MSE-minimizing weight parameter $\mu$ in Section\;\ref{sec:optPara}. 
We present our graph embedding method in Section\;\ref{sec:general}.
Experiments and conclusion are presented in Section\;\ref{sec:results} and \ref{sec:conclude}, respectively.

%% file: related.tex
We first overview graph smoothness priors in the GSP literature. 
We next discuss priors in image restoration, in particular, TV and TGV that promote PWC and PWP image reconstruction respectively.
Finally, we overview graph embedding schemes and compare them with our method to compute node coordinates for a manifold graph.

\subsection{Smoothness Priors for Graph Signal Restoration}
\label{subsec:smoothness}

Restoration of graph signals from partial and/or corrupted observations has long been studied in GSP \cite{chen15,cheung21}. 
The most common graph smoothness prior to regularize an inherently ill-posed signal restoration problem is \textit{graph Laplacian regularizer} (GLR) \cite{pang17} $\x^\top \L \x$, where $\L$ is a graph Laplacian matrix corresponding to a similarity graph kernel for signal $\x$. 
Interpreting the adjacency matrix $\W$ as a graph shift operator, \cite{chen15} proposed the \textit{graph shift variation} (GSV) $\|\x - \frac{1}{|\lambda_{\max}|} \W\x\|^2_2$, where $\lambda_{\max}$ is the spectral radius of $\W$. 
Graph signal variations can also be computed in $\ell_1$-norm as \textit{graph total variation} (GTV) \cite{elmoataz08,couprie13}.
Though convex, minimization of $\ell_1$-norm like GTV requires iterative algorithms like \textit{proximal gradient} (PG) \cite{parikh13} that are often computation-expensive.

GLR can be defined in a \textit{signal-dependent} (SD) manner as $\x^\top \L(\x) \x$, where Laplacian $\L(\x)$ is a function of sought signal $\x$.  
It has been shown \cite{pang17,liu17} that SDGLR promotes PWC signal reconstruction. 
Our work extends SDGLR to SDGGLR to promote PWP signal reconstruction, while retaining its convenient quadratic form for fast computation. 

\ifnum\arXiv=0
    Alternatively, one can model generation of graph signals probabilistically via a Gaussian process \cite{venkitaraman20,miao22} for prediction. In contrast, the aforementioned approaches are deterministic, requiring little parameterization and no training.
\fi

\subsection{Priors for Image Restoration} 

Specifically for image restoration, there exist a large variety of signal priors through decades of research \cite{milanfar13}. 
Edge-guided methods such as \textit{partial differential equations} (PDE) \cite{folland1995introduction} provide smooth image interpolation, but perform poorly when missing pixels are considerable. 
\textit{Sparse coding} \cite{elad06} that reconstructs sparse signal representations by selecting a few atoms from a trained over-complete dictionary was prevalent, but computation of sparse code vectors via minimization of $\ell_0$- or $\ell_1$-norms can be expensive. 
\textit{Total variation} (TV) \cite{chambolle1997image} was a popular image prior due to available algorithms in minimizing convex but non-differentiable $\ell_1$-norm. 
Its generalization, \textit{total generalized variation} (TGV) \cite{bredies2010total, bredies2015tgv}, better handles the known staircase effect, but retains the non-differentiable $\ell_1$-norm that requires iterative optimization. 
Like TGV, GGLR also promotes PWP signal reconstruction, but is non-convex and differentiable, leading to fast optimization.

\subsection{Graph Embeddings}

A graph embedding computes a latent coordinate $\p_i \in \mathbb{R}^K$ for each node $i$ in a graph $\cG$, so that pairwise relationships $(i,j)$ in $\cG$ are retained in the $K$-dimensional latent space as Euclidean distances.  
Existing embedding methods can be classified into three categories \cite{xu21}: matrix factorization, random walk, and deep learning. 
Matrix factorization-based methods, such as \textit{locally linear embedding} (LLE) \cite{LLE} and \textit{Laplacian eigenmap} (LE) \cite{LE}, obtain an embedding by decomposing the large sparse adjacency matrix. 
Complexities of these methods are typically $\cO(N^2)$ and thus are not scalable to large graphs.
Random walk-based methods like Deepwalk \cite{Deepwalk} and node2vec \cite{node2vec} use a random walk process to encode the co-occurrences of nodes to obtain scalable graph embeddings. 
These schemes typically have complexity $\cO(N\log{N})$.

Deep learning approaches, especially autoencoder-based methods \cite{NetWalk} and \textit{graph convolutional network} (GCN) \cite{EGCN}, are also widely studied for graph embedding. 
However, pure deep learning methods require long training time and large memory footprint to store a sizable set of trained parameters.



Our goal is narrower than previous graph embedding works in that we seek sampling coordinates for nodes in a manifold graph solely for the purpose of defining GGLR for regularization. 
To be discussed in Section\;\ref{sec:general}, a manifold graph means that each node is a uniform discrete sample of a low-dimensional continuous manifold, and thus \textit{graph embedding translates to the discovery of point coordinates in this low-dimensional manifold space}. 
This narrowly defined objective leads to an efficient method based on fast computation of extreme eigenvectors \cite{lobpcg}. 

%% file: prelim.tex
We first provide graph definitions used in our formulation, including a review of the common graph smoothness prior GLR \cite{pang17,liu17} and a notion of manifold graph \cite{pang17}.
We then review the definition of a hyperplane in $K$-dimensional Euclidean space and \textit{Gershgorin circle theorem} (GCT) \cite{varga04}.

\subsection{Graph Definitions}
\label{subsec:defn}

An undirected graph $\cG(\cV,\cE,\W)$ is defined by a set of $N$ nodes $\cV = \{1, \ldots, N\}$, edges $\cE = \{(i,j)\}$, and a real symmetric \textit{adjacency matrix} $\W \in \mathbb{R}^{N \times N}$. 
$W_{i,j} = w_{i,j} \in \mathbb{R}$ is the edge weight if $(i,j) \in \cE$, and $W_{i,j} = 0$ otherwise. 
Edge weight $w_{i,j}$ can be positive or negative; a graph $\cG$ containing both positive and negative edges is called a \textit{signed graph} \cite{su17,cheung18tsipn,yang21}. 
\textit{Degree matrix} $\D$ is a diagonal matrix with entries $D_{i,i} = \sum_{j} W_{i,j}, \forall i$. 
A \textit{combinatorial graph Laplacian matrix} $\L$ is defined as $\L \triangleq \D - \W$. 
If $\cG$ is a \textit{positive graph} (all edge weights are non-negative, \ie, $W_{i,j} \geq 0, \forall i,j$), then $\L$ is \textit{positive semi-definite} (PSD), \ie, $\x^\top \L \x \geq 0, \forall \x$ \cite{cheung18}. 

The \textit{graph spectrum} of $\cG$ is defined by the eigen-pairs of $\L$.
First, eigen-decompose $\L$ into $\L = \V \bSigma \V^{\top}$, where $\bSigma = \text{diag}(\{\lambda_k\})$ is a diagonal matrix with $N$ ordered eigenvalues, $\lambda_1 \leq \lambda_2 \leq \ldots \leq \lambda_N$, along its diagonal, and $\V = [\v_1 \ldots \v_N]$ contains corresponding eigenvectors $\v_k$'s as columns. 
The $k$-th eigen-pair, $(\lambda_k,\v_k)$, defines the $k$-th graph frequency and Fourier mode. 
$\V^{\top}$ is the \textit{graph Fourier transform} (GFT) that transforms a graph signal $\x \in \mathbb{R}^N$ to its graph frequency representation via $\tilde{\x} = \V^\top \x$ \cite{ortega18ieee}.

\subsubsection{Manifold Graph}
\label{subsubsec:manifold}

\blue{
It is common to connect a graph (a discrete object) to a continuous manifold
via the common “manifold assumption” \cite{coifman06,hein06}. 
For example, assuming that graph nodes are projected samples chosen randomly from
a uniform distribution over a manifold domain, \cite{hein06} showed that the graph Laplacian matrix converges to the continuous Laplace-Beltrami operator at an estimated rate as the number of points approaches infinity and the distances between points go to zero. 
Under the same uniform sampling assumption, \cite{hein06} showed that a Laplacian regularization term converges to a
continuous penalty functional.
}

We make a similar assumption here, namely that
a manifold graph $\cG(\cV,\cE,\W,\P)$ is a set of $N$ connected nodes that are discrete samples from a uniform distribution on a $K$-dimensional continuous manifold, each node $i \in \cV$ with coordinate $\p_i \in \mathbb{R}^K$. 
Denote by $d(i,j) \triangleq \|\p_i - \p_j\|_2$ the Euclidean distance between nodes $i$ and $j$ in the manifold space.
\blue{We assume that $\cG$ satisfies the following two properties that are common in graph embedding \cite{Cucuringu2015OrdinalEO}}:
\begin{enumerate}
\item $w_{i,j} < w_{i,k}$ implies $d(i,j) > d(i,k)$.
\item If two nodes $j$ and $k$ connected by edge $(j,k) \in \cE$ are respective $h$- and $(h+1)$-hop neighbors of node $i$ in $\cG$, then $d(i,j) < d(i,k)$. 
\end{enumerate}
The first property means that edge weight is inversely proportional to manifold distance, \ie, $w_{i,j} \propto d(i,j)^{-1}$---a reasonable assumption since an edge weight typically captures pairwise similarity / correlation. 
The second property implies that a $(h+1)$-hop node $k$ cannot be closer to $i$ than $h$-hop node $j$ if $\exists (j,k) \in \cE$. 
This is also reasonable given that hop count should reflect some notion of distance.
We invoke these properties in the DAG construction procedure in Section\;\ref{subsec:gradOp} and the graph embedding algorithm in Section\;\ref{subsec:embeddingObj}. 
Examples include $k$-nearest-neighborhood (kNN) graphs for pixels on a 2D image and for 3D points in a point cloud.

\subsubsection{Signal-Independent GLR} 
A signal $\x \in \mathbb{R}^N$ is smooth w.r.t. graph $\cG$ if its GLR, $\x^{\top} \L \x$, is small~\cite{pang17}:
\begin{align}
\x^{\top} \L \x = \sum_{(i,j) \in \cE} w_{i,j} (x_i - x_j)^2 
= \sum_k \lambda_k \tilde{x}_k^2 
\label{eq:GLR}
\end{align}
where $\tilde{x}_k$ is the $k$-th GFT coefficient of $\x$. 
In the nodal domain, a small GLR means a connected node pair $(i,j)$ has similar sample values $x_i$ and $x_j$ for large $w_{i,j}$.
In the spectral domain, a small GLR means signal energies $\tilde{x}_k^2$ mostly reside in low frequencies $\lambda_k$, \ie, $\x$ is roughly a low-pass signal. 

For a connected positive graph without self-loops, from \eqref{eq:GLR} we see that $\x^\top \L \x = 0$ iff $\x$ is a constant signal $c \1$ where $\1$ is an all-one vector, \ie, $x_i = x_j = c, \forall i,j$.
Constant signal $c \1$ is the (unnormalized) eigenvector $\v_1$ of $\L$ to first eigenvalue $\lambda_1 = 0$ (\ie, the $0$ frequency). 
Thus, signal-independent GLR (SIGLR), where graph $\cG$ is defined independent of signal $\x$, promotes constant signal reconstruction.

\subsubsection{Signal-Dependent GLR}

One can alternatively define a \textit{signal-dependent} GLR (SDGLR), where edge weights of graph $\cG$ are defined as a function of sought signal $\x$ \cite{pang17,liu17}.
Specifically, we write
\begin{align}
\x^{\top} \L(\x) \x &= 
\sum_{(i,j) \in \cE} w_{i,j}(x_i,x_j) \, (x_i - x_j)^2 
\label{eq:SDGLR} \\
w_{i,j}(x_i,x_j) &= 
\exp \left( 
- \frac{\| \f_i-\f_j \|^2_2}{\sigma_f^2}
- \frac{|x_i-x_j|^2}{\sigma_x^2} \right)
\label{eq:edgeWeight0}
\end{align}
where $\f_i \in \mathbb{R}^M$ is the \textit{feature vector} of node $i$. 
$w_{i,j}$ is \textit{signal-dependent}---it is a Gaussian function of signal difference $|x_i - x_j|$---and hence $\L(\x)$ is a function of sought signal $\x$.
This edge weight definition \eqref{eq:edgeWeight0} is analogous to \textit{bilateral filter} weights with domain and range filters \cite{tomasi98}.
Hence, each term in the sum \eqref{eq:SDGLR} is minimized when i) $|x_i - x_j|$ is very small (when $(x_i-x_j)^2$ is small), or ii) $|x_i - x_j|$ is very large (when $w_{i,j}$ is small).
Thus, minimizing GLR would force $|x_i - x_j|$ to approach either $0$ or $\infty$, promoting \textit{piecewise constant} (PWC) signal reconstruction as shown in \cite{pang17,liu17}.

We examine SDGLR more closely to develop intuition behind this PWC reconstruction \cite{bai19}. 
For notation simplicity, define $\Delta_{i,j} \triangleq (x_i - x_j)^2$ and $\gamma^\f_{i,j} \triangleq \exp \left( -\|\f_i - \f_j\|^2_2/\sigma_f^2 \right)$. 
Then, for a single term in the sum \eqref{eq:SDGLR} corresponding to $(i,j) \in \cE$, we write

\vspace{-0.05in}
\begin{small}
\begin{align}
w_{i,j}(x_i,x_j)(x_i - x_j)^2 &= \gamma^\f_{i,j} \exp \left( - \frac{\Delta_{i,j}}{\sigma_x^2} \right) \Delta_{i,j} \\
\frac{\partial w_{i,j}(x_i, x_j)(x_i - x_j)^2}{\partial \Delta_{i,j}} &= \gamma^\f_{i,j} \left(1 - \frac{\Delta_{i,j}}{\sigma_x^2} \right) \exp \left( - \frac{\Delta_{i,j}}{\sigma_x^2} \right) .
\label{eq:SDGLR_deriv}
\end{align}
\end{small}\noindent
We see that derivative \eqref{eq:SDGLR_deriv} w.r.t. $\Delta_{i,j}$ increases from $0$ till $\sigma_x^2$ then decreases, and thus $\x^\top \L(\x) \x$ is in general non-convex. 
Further, because the derivative is uni-modal, in an iterative algorithm $\Delta_{i,j}$ is promoted to either $0$ (\textit{positive diffusion}) or $\infty$ (\textit{negative diffusion}) depending on initial value \cite{pang17}. 

As an illustrative example, consider a 5-node line graph $\cG$ in Fig.\;\ref{fig:5node_graph} with initial noisy observation $\y = [2 ~2 ~1.8 ~1.2 ~1]^\top$.  
Suppose we use an iterative algorithm that alternately computes signal via $\min_{\x} \|\y-\x\|^2_2 + \mu \x^\top \L \x$ and updates edge weights $w_{i,j} = \exp(-|x_i - x_j|^2/\sigma_x^2)$ in $\L$ (initial edge weights are computed using observation $\y$).
Assume also that $\sigma_x^2 = 0.1$ and $\mu = 1$. 
It will converge to solution $\x^* = [1.92 ~1.92 ~1.92 ~1.12 ~1.12]$ after $8$ iterations; see Fig.\;\ref{fig:5node_graph}(c) for an illustration.
We see that $\x^*$ is an approximate PWC signal that minimizes SDGLR $\x^\top \L(\x) \x =  0.001$ for a (roughly) disconnected line graph\footnote{The second eigenvalue is also called the \textit{Fiedler number} that reflects the connectedness of the graph \cite{Fiedler1975APO}.}  with $\lambda_2 \approx 0$. 
This demonstrates that SDGLR does promote PWC signal reconstruction. 

\begin{figure}
\centering
\begin{tabular}{c}
\includegraphics[width=3.2in]{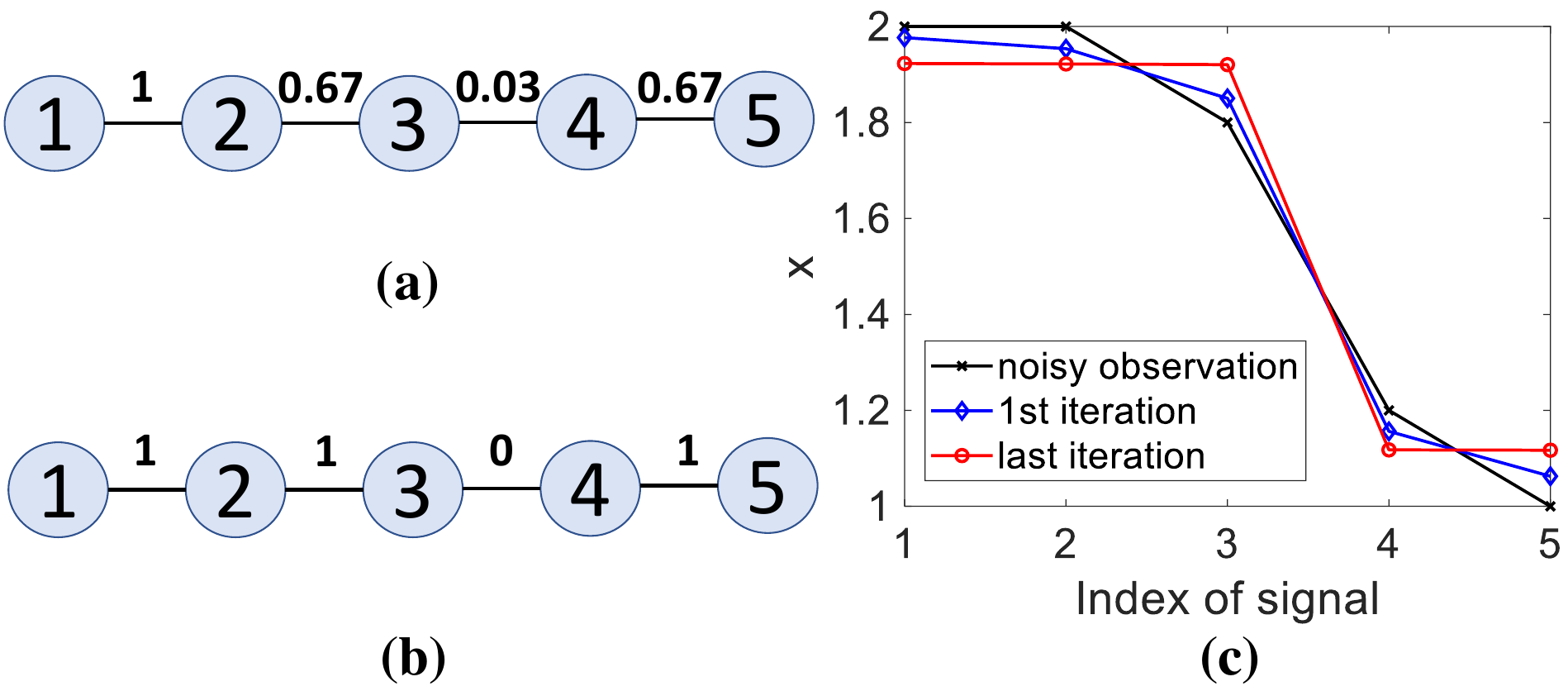}
\vspace{-0.1in}
\end{tabular}
\caption{\small Example of (a) a $5$-node line graph $\cG$ with initial edge weights, (b) the same line graph with converged edge weights, and (c) output $\x$'s during iterations, where the converged signal $\x^*$ in red is PWC.}
\label{fig:5node_graph}
\end{figure}

\begin{figure}
\centering
\begin{tabular}{c}
\includegraphics[width=3.2in]{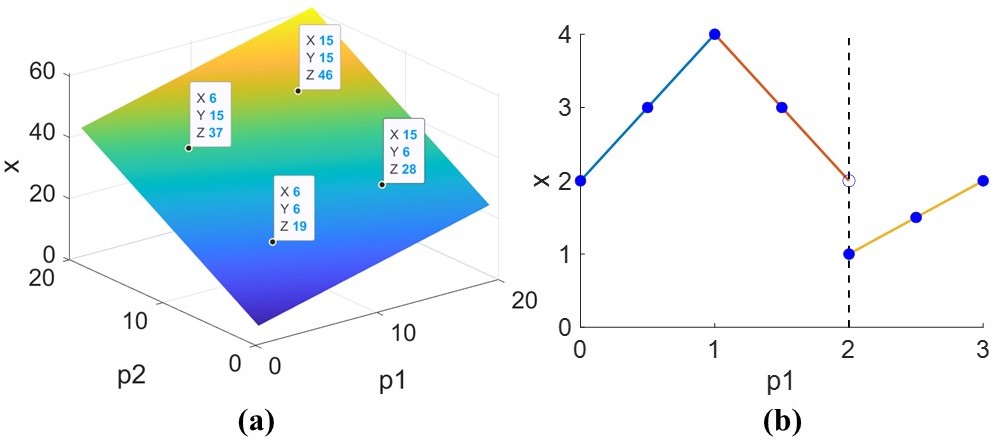}
\vspace{-0.1in}
\end{tabular}
\caption{\small Examples of (a) a planar signal $\x$ with $4$ discrete points residing on a 2D plane $x = 1 + p_1 + 2 p_2$, and (b) a 1D 3-piece linear signal \eqref{eq:3piece} with parameters $(\alpha_0^{(1)}=2, ~\alpha_1^{(1)}=2)$, $(\alpha_0^{(2)}=6, ~\alpha_1^{(2)}=-2)$, and $(\alpha_0^{(3)}=-1, ~\alpha_1^{(3)}=1)$. Note that in (b) the signal is continuous at $p_1=1$ and non-continuous at $p_1 = 2$.}
\label{fig:PlanarSignal}
\end{figure}

\subsection{$K$-dimensional Planar Signal}
\label{subsec:planar}

A discrete signal $\x \in \mathbb{R}^{N}$, where each sample $x_i$ is associated with coordinate vector $\p_i = [p_{i,1}, \ldots p_{i,K}]^\top$, is a \textit{$K$-dimensional (hyper)planar signal}, if each sample $x_i$ can be written as
\begin{align}
x_i = \alpha_0 + \sum_{k=1}^K \alpha_k p_{i,k}  
\label{eq:planar}
\end{align}
where $\{\alpha_k\}_{k=0}^K$ are parameters of the hyperplane.
In other words, all discrete points $\{(\p_i, x_i)\}_{i=1}^N$ of signal $\x$ reside on a hyperplane defined by $x = \alpha_0 + \sum_{k=1}^K \alpha_k p_k$.  
We call $\balpha = [\alpha_1 \cdots \alpha_K]^\top$ the \textit{gradient} of the hyperplane. 
For example, 
in Fig.\;\ref{fig:PlanarSignal}(a) a $2$-dimensional planar signal $\x$ has $4$ discrete points $\{(\p_i,x_i)\}_{i=1}^4$ residing on a 2D plane $x = 1 + p_1 + 2 p_2$. 

Given a $K$-dimensional planar signal $\x$ with $N > K$ points, one can compute gradient $\balpha$ via a system of linear equations. 
Specifically, using $K+1$ points $(\{(\p_i, x_i)\}_{i=1}^{K+1}$, we can write $K$ linear equations by subtracting \eqref{eq:planar} for point $(\p_i,x_i)$, $2 \leq i \leq K+1$, from \eqref{eq:planar} for point $(\p_1,x_1)$:

\vspace{-0.05in}
\begin{scriptsize}
\begin{align}
\underbrace{\left[ \begin{array}{ccc}
p_{1,1} - p_{2,1} & \ldots & p_{1,K} - p_{2,K} \\
p_{1,1} - p_{3,1} & \ldots & p_{1,K} - p_{3,K} \\
\vdots & & \vdots \\
p_{1,1} - p_{K+1,1} & \ldots & p_{1,K} - p_{K+1,K}
\end{array} \right]}_{\C} 
\left[ \begin{array}{c}
\alpha_1 \\
\alpha_2 \\
\vdots \\
\alpha_K
\end{array} \right] &= 
\underbrace{\left[ \begin{array}{c}
x_1 - x_2 \\
x_1 - x_3 \\
\vdots \\
x_1 - x_{K+1}
\end{array} \right]}_{\Delta \x} 
\nonumber \\
\C \balpha &= \Delta \x .
\label{eq:alpha_linear}
\end{align}
\end{scriptsize}\noindent
\eqref{eq:alpha_linear} produces unique $\balpha$ if \textit{coordinate matrix} $\C \in \mathbb{R}^{K \times K}$ is full column-rank;
we make this assumption in the sequel. 

\vspace{0.05in}
\noindent
\textbf{Assumption 1}: Coordinate matrix $\C$ is full column-rank.
\vspace{0.05in}

This means that a subset of $M \leq K$ points cannot reside on a $(M-2)$-dimensional hyperplane; this can be ensured via prudent sampling (see our proposed DAG construction procedure in Section\;\ref{subsec:gradOp} to select points for computation of gradient $\balpha^i$).

If $n > K+1$ points are used, then $\C \in \mathbb{R}^{(n-1) \times K}$ is a tall matrix, and gradient $\balpha$ minimizing the \textit{least square error} (LSE) between $n$ points and the best-fitted hyperplane can be obtained using \textit{left pseudo-inverse}\footnote{Left pseudo-inverse $\C^\dagger$ is well defined if $\C$ is full column-rank.} $\C^\dagger \triangleq (\C^\top \C)^{-1} \C^\top$ 
\cite{strang15}:
\begin{align}
\balpha &= \C^\dagger \Delta \x .
\end{align}

Finally, we consider the special case where points $x_i$'s are on a regularly sampled $K$-dimensional grid; for example, consider using image patch pixels on a 2D grid to compute gradient $\balpha = [\alpha_1 ~\alpha_2]^\top$ for a best-fitted 2D plane. 
In this case, computing individual $\alpha_i$'s can be done \textit{separately}. 
For example, a row of image pixels $x_i$'s have the same $y$-coordinate $p_2$, and thus using \eqref{eq:planar} one obtains $x_1 - x_i = \alpha_1 (p_{1,1} - p_{i,1})$ for each $i$, and the optimal $\alpha_1$ minimizing least square can be computed separately from $\alpha_2$ (and vice versa). 
Computing $\alpha_i$'s separately has a complexity advantage, since no computation of matrix (pseudo-)inverse of $\C$ is required.

\subsubsection{Piecewise Planar Signal}
\label{subsubsec:PWP}

\blue{
One can define a \textit{piecewise planar signal} (PWP) as a combination of planar signals in adjacent \textit{sub-domains}; similarly defined in the \textit{nodal domain theorem} \cite{Bykoglu2005NodalDT}, graph nodes $\cV$ are first partitioned into non-overlapping sub-domains of connected nodes, and two sub-domains $\cV_s$ and $\cV_t$ are adjacent if $\exists i \in \cV_s, j \in \cV_t$ such that $(i,j) \in \cE$. 
}
For example, in Fig.\;\ref{fig:PlanarSignal}(b) we see a $3$-piece 1D linear signal $x$ with $p_1 \in [0, 3]$ defined as
\begin{align}
x = \left\{ \begin{array}{ll}
\alpha_0^{(1)} + \alpha_1^{(1)} p_1 & \mbox{if}~ 0 \leq p_1 < 1 \\
\alpha_0^{(2)} + \alpha_1^{(2)} p_1 & \mbox{if}~ 1 \leq p_1 < 2 \\
\alpha_0^{(3)} + \alpha_1^{(3)} p_1 & \mbox{if}~ 2 \leq p_1 \leq 3
\end{array} \right. 
\label{eq:3piece}
\end{align}
where the first, second, and third linear pieces are characterized by parameters $(\alpha_0^{(1)} ~\alpha_1^{(1)})$, $(\alpha_0^{(2)} ~\alpha_1^{(2)})$, and $(\alpha_0^{(3)} ~\alpha_1^{(3)})$, respectively.
If one further assumes signal $x$ is \textit{continuous}, then at each boundary, the two adjacent linear pieces must coincide.
For example, $x$ is continuous at $p_1 = 1$ if $\alpha_0^{(1)} + \alpha_1^{(1)}  = \alpha_0^{(2)} + \alpha_1^{(2)}$. 
The $3$-piece linear signal in Fig.\;\ref{fig:PlanarSignal}(b) is continuous at $p_1 = 1$ but not at $p_1 = 2$. 
Signal discontinuities are common in practice such as 2D natural and depth images (\eg, boundary between a foreground object and background).
We discuss our assumption on signal continuity in Section\;\ref{subsec:PWP}.

\subsection{Gershgorin Circle Theorem}

Given a real symmetric square matrix $\M \in \mathbb{R}^{N \times N}$, corresponding to each row $i$ is a \textit{Gershgorin disc} $i$ with center $c_i \triangleq M_{i,i}$ and radius $r_i \triangleq \sum_{j \,|\, j \neq i} |M_{i,j}|$. 
By GCT \cite{varga04}, each eigenvalue $\lambda$ of $\M$ resides in at least one Gershgorin disc, \ie, $\exists i$ such that 
\begin{align}
c_i - r_i \leq \lambda \leq c_i + r_i .
\label{eq:GCT}
\end{align}
A corollary is that the smallest Gershgorin disc left-end $\lambda^-_{\min}(\M)$ is a lower bound of the smallest eigenvalue $\lambda_{\min}(\M)$ of $\M$, \ie,
\begin{align}
\lambda^-_{\min}(\M) \triangleq \min_i c_i - r_i \leq \lambda_{\min}(\M) .
\label{eq:GCT}
\end{align}
We employ GCT to compute suitable parameters for our graph embedding method in Section\;\ref{subsec:embeddingObj}.

%% file: GGLR2.tex
\begin{figure}[t]
\begin{center}
   \includegraphics[width=0.9\linewidth]{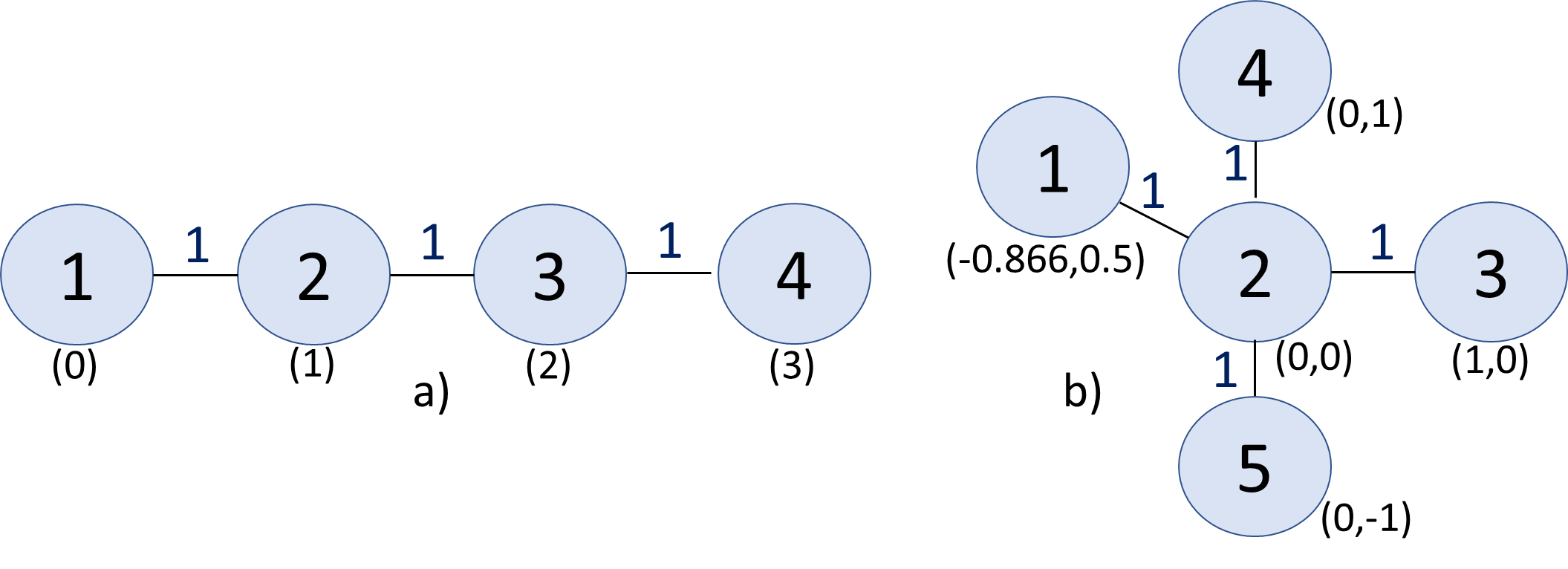}
\end{center}
\vspace{-0.2in}
\caption{Examples of a 1D 4-node graph in (a) and a 2D 5-node graph in (b). Numbers in brackets are sampling coordinates in 1D / 2D space respectively.} 
\label{fig:LapGraphEx}
\end{figure}


\subsection{Why Graph Laplacian cannot promote Planar Signals}
\label{subsec:whyNot}

As motivation, we first argue that conventional notion of ``\blue{graph gradients}", computed using Laplacian\footnote{A similar argument can be made using normalized Laplacian $\L_n$ to compute gradients for a positive connected graph; 
\ifnum\arXiv=1
    see Appendix\;\ref{append:normL}.
\else 
    see Appendix\;A in \cite{chen23manifold}.
\fi
} $\L \triangleq \D - \W$, is ill-suited to define planar graph signals.
Specifically, if $\L \x$ computes \blue{graph gradients} \cite{ortega18ieee}, then constant graph gradient across the signal---analogous to constant slope across a linear function on a regular 1D kernel---would imply a kind of linear / planar graph signal. 
However, a signal $\x$ satisfying $\L \x = c\1$ for some constant $c$ is not a general planar signal at all; $\1$ is in the \textit{left-nullspace} of $\L$ (\ie, $\L^\top \1 = \0$), which means $\1$ is not in the orthogonal \textit{column space} of $\L$, \ie, $\not\exists \x ~\mbox{s.t.}~ \L \x = \1$. 
Hence, the only signal $\x$ satisfying $\L \x = c\1$ is the constant signal $\x = c' \1$ for some $c' \neq 0$ and $c=0$.
That means a prior like $(\L\x)^\top \L (\L\x)$ 
promotes only a constant signal as $0$ frequency, no different from GLR $\x^\top \L \x$ \cite{pang17}. 

To understand why graph gradients $\L\x$ cannot define planar signals in \eqref{eq:planar}, consider a 4-node line graph with all edge weights $w_{i,i+1} = 1$ and a signal $\x = [4, 2, 0, -2]^\top$ \textit{linear} w.r.t. node indices. 
See Fig.\;\ref{fig:LapGraphEx}(a) for an illustration. 
$\x$ is linear, since the three computable slopes $(x_{i+1} - x_i)/1, i \in \{1, 2, 3\}$ are all $-2$. 
Note that a 4-node line graph only has three uniquely computable slopes, a property we will revisit in the sequel.

Computing $\L \x$ for linear $\x$, we get

\vspace{-0.05in}
\begin{small}
\begin{align*}
\L \x = \left[ \begin{array}{cccc}
1 & -1 & 0 & 0 \\
-1 & 2 & -1 & 0 \\
0 & -1 & 2 & -1 \\
0 & 0 & -1 & 1
\end{array} \right] 
\left[ \begin{array}{c}
4 \\
2 \\
0 \\
-2 
\end{array} \right] = 
\left[ \begin{array}{c}
2 \\
0 \\
0 \\
-2
\end{array} \right] 
\end{align*}
\end{small}\noindent 
which is not a constant vector.
Examining $\L \x$ closely, we see that $(\L\x)_i$ computes the difference between $x_i$ and \textit{all} its connected neighbors $x_j$ for $(i,j) \in \cE$. 
Depending on the varying number and locations of its neighbors, $(\L\x)_i$ is computing different quantities: 
$(\L\x)_1$ and $(\L\x)_4$ compute negative / positive slope, while $(\L \x)_2$ and $(\L\x)_3$ compute the \textit{difference} of slopes (discrete second derivative). 
Clearly, entries in $\L\x$ are not computing slopes $(x_{i+1} - x_i)/1$ consistently.

\blue{
Can one compute the difference of slopes for ``interior" nodes only?
First, it is difficult to define and identify ``boundary" nodes in a general graph with nodes of varying degrees.} 
But suppose ``interior" rows in $\L$ can be identified to compute $\L\x$ as discrete second derivatives for minimization, \eg,  using $\|\text{diag}(\h) \L \x\|^2_2$ as a planar signal prior, where $\text{diag}(\h)$ selects the interior entries of $\L\x$ (in the line graph example, $\h = [0, 1, 1, 0]^\top$). 
Consider the 2D 5-node example in Fig.\;\ref{fig:LapGraphEx}(b). 
Given a 2D planar signal model \eqref{eq:planar}, namely $x = p_1 + 2 p_2$, $\x = [0.134, 0, 1, 2, -2]^\top$ is planar for 2D coordinates in Fig.\;\ref{fig:LapGraphEx}(b). 
However, $(\L\x)_2 \neq 0$. 
The reason is that, unlike Fig.\;\ref{fig:LapGraphEx}(a), the coordinates of four neighboring nodes are not symmetric around interior node 2, and Laplacian $\L$ alone does not encode sufficient node location information.  
This shows that \textit{the Laplacian operator for a graph with only nodes and edges is insufficient to define planar signals in \eqref{eq:planar}}.

\subsection{Gradient Operator}
\label{subsec:gradOp}

We derive a new graph smoothness prior, \textit{gradient graph Laplacian regularizer} (GGLR), for signal restoration for the case when each graph node is endowed with a low-dimensional sampling space coordinate. 
See Table\;\ref{tab:notations} for notations.

\begin{figure}[t]
\begin{center}
\includegraphics[width=0.85\linewidth]{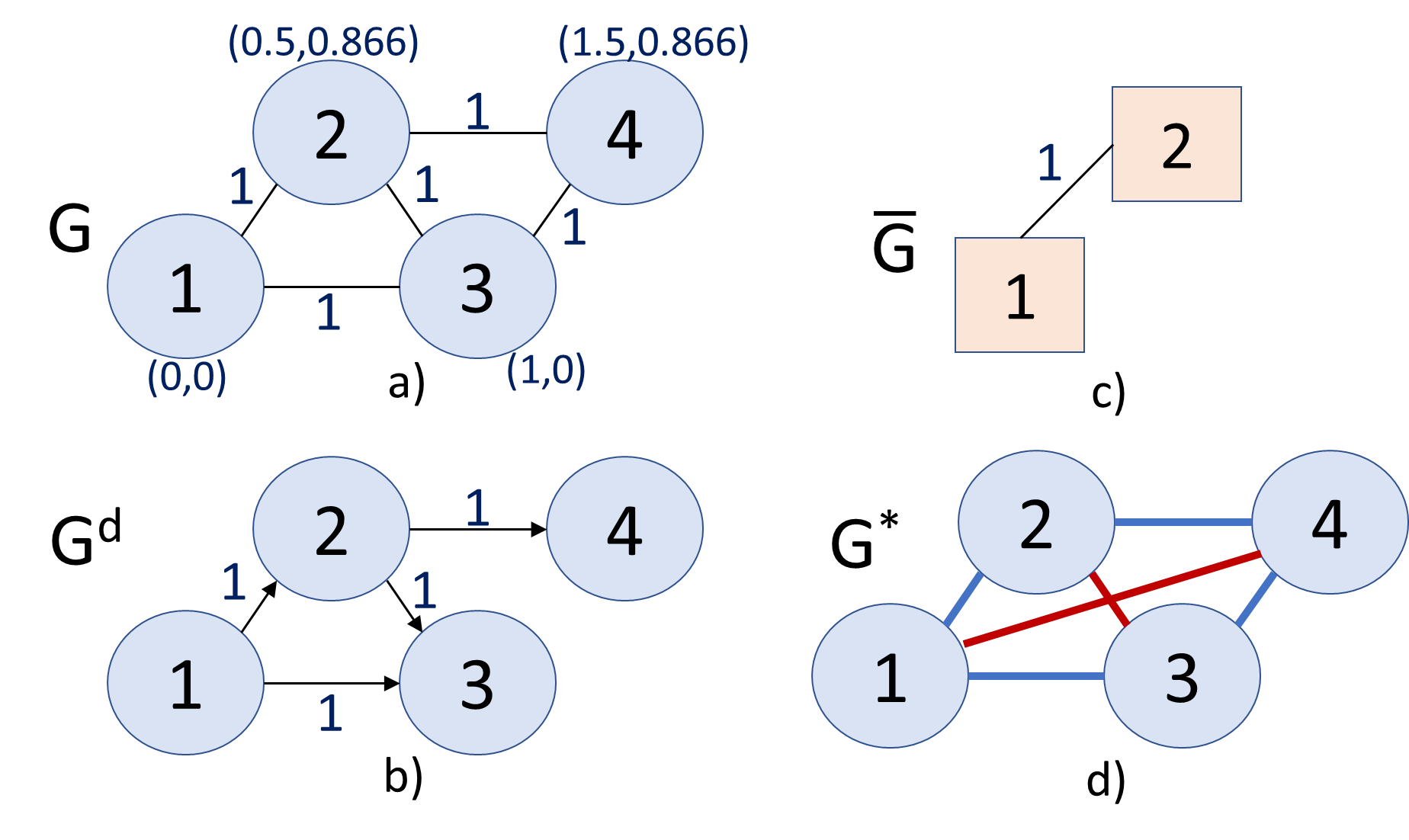}
\end{center}
\vspace{-0.25in}
\caption{An example 4-node graph $\cG$ with 2D coordinates in (a), the corresponding DAG $\cG^d$ in (b), gradient graph $\bar{\cG}$ in (c), and resulting GNG $\cG^*$ (d), where blue/red edges are positive/negative, each with magnitude $1.334$.} 
\label{fig:gradGraphEx}
\end{figure}

\begin{table}[t!]
\caption{Notations in GGLR Derivation 
}
\centering
\begin{scriptsize}
\begin{tabular}{|c|c|} \hline 
$\cG(\cV, \cE, \W)$ & graph with nodes $\cV$, edges $\cE$, adjacency matrix $\W$\\
$\p_i$ & $K$-dimensional coordinate for node $i$ \\
$K^+$ & number of out-going edges for each node in DAG \\
$\cG^d(\cV,\cE^d,\W^d)$ & DAG with nodes $\cV$, edges $\cE^d$, adjacency matrix $\W^d$ \\
$\F^i$ & gradient operator for node $i$ \\
$\W^i$ & diagonal matrix containing weights of edges $[i,j]$ \\
$\G$ & gradient computation matrix with $\{(\W^i \C^i)^\dagger \W^i \F^i\}$ \\
$\R$ & reordering matrix to define $\balpha$ \\
$\balpha^i$ & manifold gradient for node $i$ \\
$\bar{\cG}(\bar{\cV}, \bar{\cE}, \bar{\W})$ & gradient graph with nodes $\bar{\cV}$, edges $\bar{\cE}$, adj. matrix $\bar{\W}$ \\
$\bar{\L}$ & gradient graph Laplacian matrix for $\bar{\cG}$ \\
$\bar{\L}^o$ & $\text{diag}(\bar{\L}, \ldots, \bar{\L})$ to define $\cL$ \\
$\G^*(\cV, \cE^*, \W^*)$ & gradient-induced nodal graph (GNG) \\
$\cL$ & GNG Laplacian matrix \\ \hline
\end{tabular}
\end{scriptsize}
\label{tab:notations}
\end{table}

Given an undirected graph $\cG = (\cV, \cE, \W)$ with $|\cV| = N$ nodes, we assume here that each node $i \in \{1, \ldots, N\}$ is endowed with a coordinate vector $\p_i \in \mathbb{R}^K$ in $K$-dimensional sampling space, where $K \ll N$.
Examples of such graphs include kNN graphs constructed for 2D images or 3D point clouds, where the node coordinates are the sample locations in Euclidean space \cite{pang17,liu17,dinesh20}. 

For each node $i \in \cV$, we compute manifold gradient $\balpha^i$ when possible. 
First, define a \textit{directed acyclic graph} (DAG) $\cG^d(\cV,\cE^d,\W^d)$, where $K^+$ directed edges $[i,j] \in \cE^d$, $K^+ \geq K$, are drawn for each node $i$ (if possible), following the DAG construction procedure below.
Denote by $\bar{\cV} \subset \cV$ the set of nodes with $K^+$ out-going edges; \red{$|\bar{\cV}|$ is the number of computable manifold gradients.}
Given $K^+$ out-going edges $[i,j]$ stemming from node $i$, define a \textit{gradient operator} $\F^i \in \mathbb{R}^{K^+ \times N}$. 
Finally, compute gradient $\balpha^i$ for node $i \in \bar{\cV}$ using $\F^i$, weights of $K^+$ directed edges $[i,j]$, and coordinates $\p_j$ of $K^+$ connected nodes $j$'s.
$K^+$ is a parameter that trades off robustness of computed gradient (larger $K^+$) with gradient locality (smaller $K^+$).

\vspace{0.05in}
\noindent
\textbf{DAG Construction Procedure}:
\begin{enumerate}
\item Initialize a \textit{candidate list} $\cS$ with 1-hop neighbor(s) $j$ of node $i$ satisfying the \textit{acylic condition}: ~$\exists k^o \geq 1$ such that
\begin{enumerate}[(1)]
\item $p_{j,k^o} > p_{i,k^o}$, and
\item $p_{j,k} = p_{i,k}$, for all $k$ s.t. $k^o > k \geq 1$.
\end{enumerate} 
\item If the number of directed edges from $i$ is $< K^+$ and $\cS \neq \emptyset$, then
\begin{enumerate}
\item Construct edge $[i,j^*]$ from $i$ to the closest node $j^*$ in $\cS$ in Euclidean distance with weight $w^d_{i,j^*} \triangleq \prod_{(s,t) \in \cP_{i,j^*}} w_{s,t}$, where $\cP_{i,j^*}$ is the shortest path from $i$ to $j^*$. 
Remove $j^*$ from list $\cS$.
\item Add 1-hop neighbor(s) $l$ of $j^*$ that satisfy the acylic condition, and are not already selected for connection to $i$, to list $\cS$. Repeat Step 2.
\end{enumerate}
\end{enumerate}

\vspace{0.05in}
After the procedure, the $K^+$ identified directed edges $\{[i,j]\}$ are drawn for $\cG^d$ and compose $\cE^d$, with computed edge weights $\{w^d_{i,j}\}$ composing weight matrix $\W^d$. 
Fig.\;\ref{fig:gradGraphEx}(a) shows an example $4$-node graph with 2D coordinates in brackets. 
The corresponding graph $\cG^d$ with directed edges drawn using the DAG construction procedure is shown in Fig.\;\ref{fig:gradGraphEx}(b). 
In this example, only two manifold gradients $\balpha^1$ and $\balpha^2$ are computable, \ie, $\bar{\cV} = \{1, 2\}$. 

The acylic condition ensures that $\cE^d$ constitute a graph $\cG^d$ without cycles.

\begin{lemma}
A graph constructed with directed edges satisfying the acylic condition has no cycles.
\label{lemma:acyclic}
\end{lemma}\noindent
See Appendix\;\ref{append:acyclic} for a proof.

\vspace{0.05in}
\red{
Typically, the construction procedure results in a connected DAG $\cG^d$ with a single \textit{root node} (with no in-coming edges), and $N-K^+$ nodes have $K^+$ out-going edges. For example, for $K^+=K=1$, $N-1$ nodes have one out-going edge each.
We assume this for the constructed DAG in the sequel.
}

\vspace{0.05in}
\noindent
\textbf{Assumption 2}: DAG $\cG^d$ is a connected graph with a single root node, where $|\bar{\cV}|K \geq N-1$.

\vspace{0.05in}
Using the two properties of a manifold graph discussed in Section\;\ref{subsubsec:manifold}, one can prove the optimality of constructed directed edges $[i,j]$ for each node $i \in \red{\bar{\cV}}$.  
\begin{lemma}
For each node $i \in \bar{\cV}$, the DAG construction procedure finds the closest $K^+$ nodes $j$ to node $i$ in Euclidean distance satisfying the acyclic condition.
\label{lemma:localGradient}
\end{lemma}\noindent 
See Appendix\;\ref{append:locality} for a proof.

\vspace{0.05in}
\noindent
\textbf{Remarks}:
Just like the 4-point line graph in Fig.\;\ref{fig:LapGraphEx}(a) having only three unique computable slopes, the acyclic nature of the constructed directed graph (ensured by Lemma\;\ref{lemma:acyclic}) guarantees that each computed gradient $\balpha^i$ is distinct.
Lemma\;\ref{lemma:localGradient} ensures that the computed gradient $\balpha^i$ is as local to node $i$ as possible, given $K^+$ closest neighbors satisfying the acyclic condition are used in the gradient computation.


Denote by $\cE^d(i,k)$ the designation node of the $k$-th directed edge stemming from node $i$ in $\cG^d$.
Given DAG $\cG^d$, we can define entries in gradient operator $\F^i \in \mathbb{R}^{K^+ \times N}$ as
\begin{align}
F^i_{m,n} &= \left\{ \begin{array}{ll}
1 & \mbox{if}~ n = i \\
-1 & \mbox{if}~ n = \cE^d(i,m) \\
0 & \mbox{o.w.}
\end{array} \right. .
\label{eq:gradOp}
\end{align}

Continuing our $4$-node example graph in Fig.\;\ref{fig:gradGraphEx}(a), gradient operators $\F^1, \F^2 \in \mathbb{R}^{2 \times 4}$ corresponding to $\cG^d$ in (b) are

\vspace{-0.05in}
\begin{small}
\begin{align*}
\F^1 = \left[ \begin{array}{cccc}
1 & -1 & 0 & 0 \\
1 & 0 & -1 & 0
\end{array} \right], ~~
\F^2 = \left[ \begin{array}{cccc}
0 & 1 & -1 & 0 \\
0 & 1 & 0 & -1
\end{array} \right] .
\label{eq:ex_F}
\end{align*}
\end{small}

\subsection{Gradient Graph}

Accompanying $\F^i$ is a \textit{coordinate matrix} $\C^i \in \mathbb{R}^{K^+ \times K}$ defined as
\begin{align}
C^i_{m,n} = 
p_{i,n} - p_{\cE^d(i,m),n} .
\end{align}
In \eqref{eq:alpha_linear}, $\C$ is an example coordinate matrix.
Continuing our example graph in Fig.\;\ref{fig:gradGraphEx}(a), $\C^1$ corresponding to $\F^1$ is

\vspace{-0.05in}
\begin{small}
\begin{align*}
\C^1 = \left[ \begin{array}{cc}
p_{1,1} - p_{2,1} & p_{1,2} - p_{2,2} \\
p_{1,1} - p_{3,1} & p_{1,2} - p_{3,2}
\end{array} \right] =
\left[ \begin{array}{cc}
-0.5 & -0.866 \\
-1 & 0
\end{array} \right] .
\end{align*}
\end{small}

Given input signal $\x \in \mathbb{R}^N$, we can write linear equations for gradient $\balpha^i = [\alpha_{i,1} \cdots \alpha_{i,K}]^\top \in \mathbb{R}^K$ for node $i \in \bar{\cV}$, similar to \eqref{eq:alpha_linear} for a $K$-dimensional planar signal:
\begin{align}
\C^i \balpha^i = \F^i \x . 
\label{eq:alpha}
\end{align}
If $K^+ = K$ and $\C^i$ is full rank, then $\balpha^i$ can be computed as $\balpha^i = (\C^i)^{-1} \F^i \x$.

If $K^+ > K$, then $\C^i$ is a tall matrix, and we compute $\balpha^i$ that minimizes the following \textit{weighted least square} (WLS):
\begin{align}
\min_{\balpha^i} \| \W^i (\F^i \x - \C^i \balpha^i) \|^2_2  ,  
\label{eq:WLS}
\end{align}
where $\W^i = \text{diag}(\{w^d_{i,\cE^d(i,m)}\}_{m=1}^{K^+})$ is a diagonal matrix containing weights of edges connecting nodes $i$ and $\cE^d(i,m)$ as diagonal terms.
\eqref{eq:WLS} states that square errors for nodes $\cE^d(i,m)$ are weighted according to edge weights.  
Solution to \eqref{eq:WLS} is
\begin{align}
\balpha^i = (\W^i \C^i)^\dagger \W^i \F^i \x    
\label{eq:alphai}
\end{align}
where $(\W^i\C^i)^\dagger$ is the left pseudo-inverse of $\W^i\C^i$. 
By Assumption 1 $\C^i$ is full column-rank, hence $\W^i\C^i$ is also full column-rank, and thus $(\W^i\C^i)^\dagger$ is well defined.

Given computed gradient $\balpha^i$, we next construct an undirected \textit{gradient graph} $\bar{\cG}(\bar{\cV},\bar{\cE},\bar{\W})$: 
connect node pair $i,j \in \bar{\cV}$ with an undirected edge if $\exists (i,j) \in \cE$.
Edge $(i,j)$ has positive weight $\bar{w}_{i,j} \in \mathbb{R}^+$, resulting in a positive undirected graph $\bar{\cG}$, with a Laplacian $\bar{\L}$ that is provably PSD \cite{cheung18}. 

To promote planar signal reconstruction, we retain the same edge weights as the original graph $\cG$, \ie, $\bar{w}_{i,j} = w_{i,j}, \forall (i,j) \in \bar{\cE}$.
Alternatively, one can define $\bar{w}_{i,j}$ in a \textit{signal-dependent} way, as a function of gradients $\balpha^i$ and $\balpha^j$, \ie, 
\begin{align}
\bar{w}_{i,j} = \exp \left( -
\frac{\|\balpha^i - \balpha^j\|^2_2}{\sigma_\alpha^2}
\right) , 
\label{eq:gradWeight}
\end{align}
where $\sigma_\alpha > 0$ is a parameter.  
We discuss PWP in detail in Section\;\ref{subsec:PWP}. 

See Fig.\;\ref{fig:gradGraphEx}(c) for the gradient graph $\bar{\cG}$ corresponding to the 4-node DAG $\cG^d$. 
Using $\bar{\cG}$, one can define a PSD \textit{gradient graph Laplacian matrix} $\bar{\L}$ using definitions in Section\;\ref{sec:prelim}. 
Continuing our 4-node graph example in Fig.\;\ref{fig:gradGraphEx}(a), the gradient graph Laplacian $\bar{\L}$ in this case is
\begin{align}
\bar{\L} = \left[ \begin{array}{cc}
1 & -1 \\
-1 & 1
\end{array} \right] \succeq 0 .
\end{align}

For later derivation, we collect gradients $\balpha^i$ for all $i \in \bar{\cV}$ into vector $\balpha \in \mathbb{R}^{|\bar{\cV}|K}$, written as
\begin{align}
\balpha &= \R
\underbrace{\left[ \begin{array}{c}
(\W^1\C^1)^{\dagger} \W^1 \F^1 \\
\vdots \\
(\W^{|\bar{\cV}|} \C^{|\bar{\cV}|})^{\dagger} \W^{|\bar{\cV}|} \F^{|\bar{\cV}|}
\end{array} \right]}_{\G} 
\x 
\label{eq:alpha_vec}
\end{align}
where $\G \in \mathbb{R}^{|\bar{\cV}|K \times N}$ is a gradient computation matrix, and $\R \in \{0,1\}^{|\bar{\cV}|K \times |\bar{\cV}|K}$ is a reordering matrix so that $\balpha = [\alpha^1_1, \ldots, \alpha^{|\bar{\cV}|}_1, \ldots, \alpha^1_K, \ldots, \alpha^{|\bar{\cV}|}_K]^\top$.
Continuing our $4$-node example graph in Fig.\;\ref{fig:gradGraphEx}(a), $\R$ is
\begin{align}
\R = \left[ \begin{array}{cccc}
1 & 0 & 0 & 0 \\
0 & 0 & 1 & 0 \\
0 & 1 & 0 & 0 \\
0 & 0 & 0 & 1
\end{array} \right] .
\end{align}

\subsection{Gradient Graph Laplacian Regularizer}

We can now define GGLR $\text{Pr}(\balpha)$ for $\balpha$ as follows.
Define first $\bar{\L}^o = \text{diag}(\bar{\L}, \ldots, \bar{\L}) \in \mathbb{R}^{|\bar{\cV}|K \times |\bar{\cV}|K}$ that contains $K$ Laplacians $\bar{\L}$ along its diagonal. 
We write
\begin{align}
\text{Pr}(\balpha) &= \balpha^\top \bar{\L}^o \balpha = \left( \R \G \x \right)^{\top} \bar{\L}^o \R \G \x 
\label{eq:GGLR} \\
&= \x^{\top} \underbrace{\G^\top \R^\top \bar{\L}^o \R \G}_{\cL} \x 
= \x^{\top} \cL \x ,
\label{eq:GGLR}
\end{align}
where $\cL$ is a graph Laplacian matrix corresponding to a \textit{gradient-induced nodal graph} (GNG) $\cG^*$ computed from $\bar{\L}^o$.

Continuing our $4$-node example graph in Fig.\;\ref{fig:gradGraphEx}(a), GNG Laplacian $\cL$ is computed as
\begin{align*}
\cL = \G^{\top} \R^\top \bar{\L}^o \R \G = 1.334 \left[ \begin{array}{cccc}
1 & -1 & -1 & 1 \\
-1 & 1 & 1 & -1 \\
-1 & 1 & 1 & -1 \\
1 & -1 & -1 & 1
\end{array} \right] .
\end{align*}

The corresponding GNG $\cG^*$ is shown in Fig.\;\ref{fig:gradGraphEx}(d).
We see that $\cG^*$ is a \textit{signed} graph with a symmetric structure of positive / negative edges, all with magnitude $1.334$. 
Despite the presence of negative edges, $\cL = \G^{\top} \R^\top \bar{\L} \R \G$ is PSD and provably promotes planar signal reconstruction\blue{---\ie, planar signals $\x$ are the $0$ frequencies of $\cL$ and compute to $\x^\top \cL \x = 0$}.
We state the general statement formally below.

\begin{theorem}
\blue{Gradient-induced nodal graph (GNG) Laplacian $\cL$ is PSD with all planar signals as the $0$ frequencies, and thus GGLR $\x^\top \cL \x$ in \eqref{eq:GGLR} promotes planar signal reconstruction, \ie, planar signals $\x$ compute to $\x^\top \cL \x = 0$.} 
\label{thm:GGLR}
\end{theorem}\noindent
See Appendix\;\ref{append:GGLR} for a proof.

\vspace{0.1in}
\noindent
\textbf{Remarks}:
Laplacian matrices for signed graphs were studied previously. 
\cite{su17} studied GFT for signed graphs, where a self-loop of weight $2|w_{i,j}|$ was added to each endpoint of an edge $(i,j)$ with negative weight $w_{i,j} < 0$.
By GCT, this ensures that the resulting graph Laplacian is PSD. 
\cite{cheung18tsipn} ensured Laplacian $\L$ for a signed graph is PSD by shifting up $\L$'s spectrum to $\L + |\lambda_{\min}^-|\, \I$, where $\lambda_{\min}^-$ is a fast lower bound of $\L$'s smallest (negative) eigenvalue.
\cite{yang21} proved that the Gershgorin disc left-ends of Laplacian $\L$ of a \textit{balanced} signed graph (with no cycles of odd number of negative edges) can be aligned at $\lambda_{\min}$ via a similarity transform.
This means that the smallest disc left-end of $\L$, via an appropriate similarity transform, can served as a tight lower bound of $\lambda_{\min}$ to test if $\L$ is PSD.
Theorem\;\ref{thm:GGLR} shows yet another construction of a signed graph
with corresponding Laplacian that is provably PSD.

\red{We study the spectral properties of GNG Laplacians next.}

\subsection{Spectral Properties of GNG Laplacian}

We next prove the following spectral properties of GNG Laplacian $\cL$. 
The first lemma establishes the eigen-subspace for eigenvalue $0$.
The next lemma defines the dimension of this eigen-subspace for eigenvalue $0$.

\red{
\begin{lemma}
Corresponding to eigenvalue $\lambda_1 = 0$, GNG Laplacian $\cL$ has (non-orthogonal) eigenvectors $\{\u_k\}_{k=0}^K$, where $\u_0 = \1$ and $\u_k = [\p_{1,k}, \ldots, \p_{N,k}]^\top, k \in \{1, \ldots, K\}$.
\label{lemma:eVec}
\end{lemma}
}
\begin{proof}
\red{
By definition of $\cL$, $\cL \1 = \G^\top \R^\top \bar{\L}^o \R \G \, \1$. 
Thus,  
\begin{align}
\G \1 &= \left[ \begin{array}{c}
(\W^1 \F^1 \P)^\dagger \W^1 \F^1 \\
\vdots \\
(\W^{|\Bar{\cV}|} \F^{|\Bar{\cV}|} \P)^\dagger \W^{|\Bar{\cV}|} \F^{|\Bar{\cV}|}
\end{array}
\right] \1 \stackrel{(a)}{=} \0
\end{align}
where $(a)$ follows since $\F^i \1 = 0, \forall i$.
Thus, $\u_0 = \1$ is an eigenvector corresponding to $\lambda_1 = 0$. 
}

\red{
Next, we note that $\C^i = \F^i \P$, where $\P \triangleq [\p_1^\top; \ldots; \p_N^\top] \in \mathbb{R}^{N \times K}$ is a matrix containing the $K$-dimensional coordinates of $N$ graph nodes as rows: 
\begin{align}
\C^i &= \F^i \P = \F^i \left[ \begin{array}{c}
\p_1^\top \\
\vdots \\
\p_N^\top
\end{array} \right] .
\end{align}
We can now write
\begin{align}
\R \G \u_k &\stackrel{(a)}{=} \R \left[ \begin{array}{c}
(\W^1 \F^1 \P)^\dagger \W^1 \F^1 \\
\vdots \\
(\W^{|\Bar{\cV}|} \F^{|\Bar{\cV}|} \P)^\dagger \W^{|\Bar{\cV}|} \F^{|\Bar{\cV}|}
\end{array}
\right] \P_k
\nonumber \\
&\stackrel{(b)}{=} \R 
\left[ \begin{array}{c}
\e_k \\
\vdots \\
\e_k 
\end{array} \right]
\stackrel{(c)}{=} \left[ \begin{array}{c}
\0_{(k-1)|\bar{\cV}|} \\
\1_{|\bar{\cV}|} \\
\0_{(K-k)|\bar{\cV}|}
\end{array}
\right]
\end{align}
where $\1_n$ ($\0_n$) is a length-$n$ vector of all ones (zeroes), and $\e_k$ is the length-$K$ canonical column vector of all zeros except for the $k$-th entry that equals to $1$. 
$(a)$ follows since $\u_k = \P_k$, the $k$-th column of $\P$;
$(b)$ follows since $(\W^j\F^j\P)^\dagger$ is the left pseudo-inverse of $\W^j\F^j\P$, $\forall j$;
$(c)$ follows since $\R$ reorders entries of manifold gradients $\balpha^i$'s. 
We see that
\begin{align}
\cL \u_k &= \G^\top \R^\top \bar{\L}^o \R \G \u_k  
\\
&= \G^\top \R^\top 
\text{diag}(\bar{\L}, \ldots, \bar{\L}) 
\left[ \begin{array}{c}
\0_{(k-1)|\bar{\cV}|} \\
\1_{|\bar{\cV}|} \\
\0_{(K-k)|\bar{\cV}|}
\end{array}
\right] \stackrel{(a)}{=} \0_{K|\bar{\cV}|}
\nonumber 
\end{align}
where $(a)$ follows since $\1_{|\bar{\cV}|}$ is the (unnormalized) eigenvector of $\bar{\L}$ corresponding to eigenvalue $0$. 
}
\end{proof}

\red{
\begin{lemma}
GNG Laplacian $\cL$ has eigenvalue $\lambda_1 = 0$ with multiplicity $K+1$.     
\label{lemma:eVal}
\end{lemma}
\begin{proof}
Given $\cL = \G^\top \R^\top \bar{\L}^o \R \G$, where $\R$ is full-rank, $\cL \x = \0$ if i) $\G \x = \0$, or ii) $\bar{\L}^o \R \G \x = \0$.  
Given $\G = [ (\W^1 \C^1)^\dagger \W^1 \F^1; \ldots; (\W^{|\bar{\cV}|} \C^{|\bar{\cV}|})^\dagger \W^{|\bar{\cV}|} \F^{|\bar{\cV}|} ]$, define $\F = [\F^1; \ldots ; \F^{|\bar{\cV}|}] \in \{1,0\}^{|\bar{\cV}|K^+ \times N}$. 
Since $\F^i \1 = \0, \forall i$, $\F \1 = \0$. 
Further, by Assumption 2, the constructed DAG connects all $N$ nodes with a single root. 
Hence, corresponding to each node $i$ of $N-1$ non-root nodes is at least one row in $\F$ where the $i$-th entry is $-1$ with all $0$'s to the right.
Thus, with \textit{topological sorting} \cite{clr} of rows and columns in $\F$, $\F^\top$ would be in \textit{row echelon form} \cite{golub12} for all rows corresponding to $N-1$ non-root nodes.
For example, $\F$ and $\F^\top$ for the four-node graph in Fig.\;\ref{fig:gradGraphEx} are
\begin{align*}
\F = \left[ \begin{array}{cccc}
1 & -1 & 0 & 0 \\
1 & 0 & -1 & 0 \\
0 & 1 & -1 & 0 \\
0 & 1 & 0 & -1 \\
\end{array} \right], ~
\F^\top = \left[ \begin{array}{cccc}
1 & 1 & 0 & 0 \\ \hline
-1 & 0 & 1 & 1 \\
0 & -1 & -1 & 0 \\
0 & 0 & 0 & -1
\end{array} \right]
\end{align*}
This means $\F$ has rank $N-1$, and thus the dimension of $\F$'s nullspace $\cN(\F)$ is $1$ with $\1$ the lone vector spanning $\cN(\F)$.
By definition, $\bar{\L}^o = \text{diag}(\bar{\L}, \ldots, \bar{\L})$, where $\bar{\L}$ is a combinatorial Laplacian for a positive connected gradient graph without self-loops.  
Hence, each $\bar{\L}$ has eigenvalue $0$ of multiplicity $1$, and $\bar{\L}^o$ has eigenvalue 0 with multiplicity $K$, where the corresponding eigenvectors are $\{\u_k\}_{k=1}^K$ defined in Lemma\;\ref{lemma:eVec}. 
We conclude that $\cL$ has eigenvalue $0$ with multiplicity $K+1$.
\end{proof}
}

\noindent
\textbf{Remarks:}
$\{\u_k\}$ are mutually linearly independent if the coordinate matrix $\C$ is full column-rank---Assumption 1 discussed in Section\;\ref{subsec:planar}.
$\{\u_k\}$ can be orthonormalized via the known Gram-Schmidt procedure \cite{golub12}.  

\red{
One important corollary is that the eigen-subspace for eigenvalue $0$ is the space of planar signals.
\begin{corollary}
Signal $\x$ that is a linear combination of eigenvectors $\{\u_k\}_{k=0}^{K}$ spanning $K+1$-dimensional eigen-subspace corresponding to eigenvalue $0$ is a planar signal.    
\end{corollary}
\begin{proof}
We show that signal $\x = \sum_{k=0}^K a_k \u_k$ for coefficients $\{a_k\}$ has constant gradient, and hence is planar.
By \eqref{eq:alphai}, gradient $\balpha^i$ for each node $i$ in gradient graph $\bar{\cG}$ is
\begin{align}
\balpha^i &= (\W^i\C^i)^\dagger \W^i \F^i \x \stackrel{(a)}{=} (\W^i\F^i \P)^\dagger \W^i \F^i \sum_{k=0}^K a_k \u_k
\nonumber \\
&\stackrel{(b)}{=} \sum_{k=1}^K a_k (\W^i\F^i \P)^\dagger \W^i \F^i \P_k \stackrel{(c)}{=} \sum_{k=1}^K a_k \e_k .
\end{align}
$(a)$ follows since $\C^i = \F^i \P$; $(b)$ follows since $\F^i \u_0 = \0$ and $\u_k = \P_k, \forall k \geq 1$; $(c)$ follows since $(\W^i\F^i\P)^\dagger$ is the left pseudo-inverse of $\W^i\F^i\P$.
We see that gradient $\balpha^i = \sum_k a_k \e_k$ is the same for all $i$, and hence $\x$ is a planar signal. 
Indeed, $\x$ computes to $\x^\top \cL \x = 0$, since each basis component $\u_k$ computes to $\u_k^\top \cL \u_k = 0$. 
\end{proof}
}

\red{
Another corollary is that GGLR defined using GNG Laplacian $\cL$ is invariant to constant shift.
\begin{corollary}
GGLR is invariant to constant shift, \ie, $(\x + c\1)^\top \cL (\x + c\1) = \x^\top \cL \x, \forall \x \in \mathbb{R}^N$.    
\label{lemma:cShift}
\end{corollary}
\begin{proof}
We first write
\begin{align}
(\x + c\1)^\top \cL (\x + c\1) &= \x^\top \cL \x + 2c \, \x^\top \cL \1 + c^2 \1^\top \cL \1 
\nonumber \\
&\stackrel{(a)}{=} \x^\top \cL \x
\end{align}
where $(a)$ follows since $\cL \1 = \0$ by Lemma\;\ref{lemma:eVec}. 
We conclude that $(\x + c\1)^\top \cL (\x + c\1) = \x^\top \cL \x$.
\end{proof}
}

\subsection{Piecewise Planar Signal Reconstruction}
\label{subsec:PWP}

Unlike SDGLR in \cite{pang17,liu17} that promotes PWC signal reconstruction, we seek to promote PWP signal reconstruction via signal-dependent GGLR (SDGGLR).
We first demonstrate that gradient graph weight assignment \eqref{eq:gradWeight} promotes \textit{continuous} PWP reconstruction by extending our analysis for SDGLR in Section\;\ref{subsec:defn}. 
First, define \textit{square difference of manifold gradient} $\delta_{i,j} \triangleq \|\balpha^i - \balpha^j\|^2$.
Signal-dependent edge weight \eqref{eq:gradWeight} is thus $\bar{w}_{i,j} = \exp \left( - \delta_{i,j} / \sigma_\alpha^2 \right) \geq 0$. 
Writing GGLR \eqref{eq:GGLR} as a sum
\begin{align}
\balpha^\top \bar{\L}^o \balpha &= \sum_{(i,j) \in \bar{\cE}} \exp \left( - \frac{\delta_{i,j}}{\sigma_\alpha^2}\right) \delta_{i,j}, 
\end{align}
we see that each term in the sum corresponding to $(i,j) \in \bar{\cE}$ goes to $0$ if either i) $\delta_{i,j} \approx 0$, or ii) $\delta_{i,j} \approx \infty$ (in which case edge weight $\bar{w}_{i,j} \approx 0$). 
Thus, minimizing GGLR \eqref{eq:GGLR} iteratively would converge to a solution $\x^*$ where clusters of connected node-pairs in $\bar{\cG}$ have very small $\delta_{i,j}$, while connected node-pairs across clusters have very large $\delta_{i,j}$. 

As an illustrative example, consider a 5-node line graph $\cG$ with initial noisy signal $\y = [2 ~2.8 ~3.1 ~2.5 ~1.2]^\top$.
Setting $\sigma_g^2 = 0.5$, the corresponding gradient graph $\bar{\cG}$ and GNG are shown in Fig.\;\ref{fig:gradient_graph}(a) and (b), respectively. 
Using an iterative algorithm that alternately computes a signal by minimizing objective $\min_{\x} \|\y-\x\|^2_2 + \mu \x^\top \cL \x$ and updates edge weights $\bar{w}_{i,j} = \exp \left( - \delta_{i,j} / \sigma_g^2 \right)$ in $\bar{\L}$, where $\mu = 0.25$, it converges to solution $\x^* = [2.08 ~ 2.66 ~3.25 ~2.29 ~1.32]$ after $6$ iterations. 
The corresponding gradient graph and GNG are shown in Fig.\;\ref{fig:gradient_graph}(c) and (d), respectively.
In Fig.\;\ref{fig:gradient_graph}(f), we see that the converged signal follows two straight lines that coincide at their boundary (thus continuous). 
In Fig.\;\ref{fig:gradient_graph}(e), neighboring nodes have the same gradients except node-pair $(2,3)$.
The converged signal $\x^*$ that minimizes SDGGLR $\x^\top \cL (\x) \x =  1.7e-04$ is PWP.

\begin{figure}
\centering
\begin{tabular}{cc}
\includegraphics[width=3.4in]{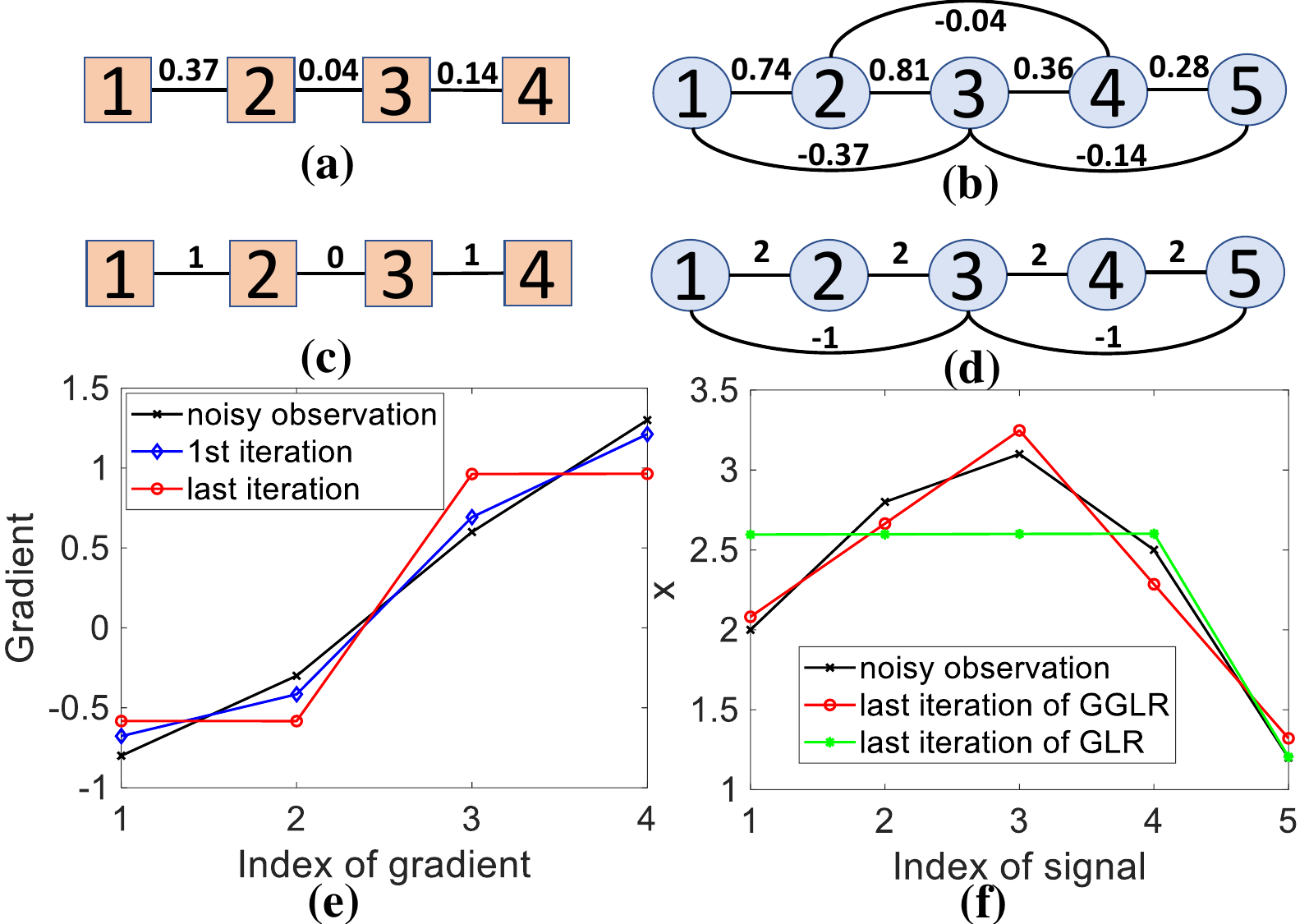}
\vspace{-0.05in}
\end{tabular}
\caption{\small (a) An example gradient graph $\bar{\cG}$ with edge weights initialized from a noisy $5$-node signal on a line graph. 
(b) A corresponding GNG $\cG^*$ to the gradient graph in (a).
(c) A gradient graph with edge weights converged after $6$ iterations.
(d) A corresponding GNG to the gradient graph in (c). 
(e) Computed gradients $\balpha^i$'s after different numbers of iterations. 
(f) Computed signals $\x$'s using GLR or GGLR for regularization after different numbers of iterations.  }
\label{fig:gradient_graph}
\end{figure}

We further generalize and consider also \textit{non-continuous} PWP signals; 
see Fig.\;\ref{fig:PlanarSignal}(b) for an example where line over domain $[1,2)$ does not coincide with line over $[2,3]$ at boundary $p_1 = 2$. 
Non-continuous PWP signals include PWC signals as a special case, and thus they are more general than continuous PWP signals. 
The difficulty in promoting non-continuous PWP reconstruction is that a discrete gradient computed across two neighboring pieces---called \textit{false gradient}---is not a true gradient of any one individual piece, and thus should be removed from the GGLR calculation.
In Fig.\;\ref{fig:PlanarSignal}(b), given neighboring points $(1.5,3)$ and $(2,1)$ are from two separate linear pieces, gradient computed across them is a false gradient. 

We detect and handle false gradients to promote non-continuous PWP reconstruction as follows.
First, we assume that the given signal is continuous PWP and perform SDGGLR in a small number of iterations, updating gradient edge weights using \eqref{eq:gradWeight} in the process. 
Second, we set a threshold (\eg, twice the average of the computed gradient norm) to identify false gradients and remove corresponding nodes in the gradient graph.
Finally, we perform iterative SDGGLR on the modified gradient graph until convergence. 
We show the importance of detecting and removing false gradients in Section\;\ref{sec:results}.

\subsection{MAP Optimization}
\label{subsec:MAP}

We now formulate a \textit{maximum a posteriori} (MAP) optimization problem using a linear signal formation model and GGLR \eqref{eq:GGLR} as signal prior, resulting in
\begin{align}
\min_\x \|\y - \H \x\|^2_2 + \mu \x^\top \cL \x
\label{eq:MAP_general}
\end{align}
where $\y \in \mathbb{R}^M$ is an observation vector, $\H \in \mathbb{R}^{M \times N}$, $M \leq N$, is the signal-to-observation mapping matrix, and $\mu >0 $ is a tradeoff parameter (to be discussed in details in Section\;\ref{sec:optPara}). 
For example, $\H$ is an identity matrix if \eqref{eq:MAP_general} is a denoising problem \cite{pang17}, $\H$ is a long $0$-$1$ sampling matrix if \eqref{eq:MAP_general} is an interpolation problem \cite{chen21}, and $\H$ is a low-pass blur filter if \eqref{eq:MAP_general} is a deblurring problem \cite{bai19}. 
We discuss specific applications of \eqref{eq:MAP_general} in Section\;\ref{sec:results}.

\eqref{eq:MAP_general} has an easily computable solution if $\H$ satisfies two conditions, which we state formally below.

\begin{theorem}
Optimization \eqref{eq:MAP_general} has a closed-form solution: 
\begin{align}
\x^* = { \underbrace{ \left( \H^\top \H + \mu \cL \right) }_{\bPhi} }^{-1} \H^\top \y 
\label{eq:MAP_sol_general}
\end{align}
if 
i) $\H\u_k \neq \0, \forall k$, 
and ii) $\{\H \u_k\}_{k=0}^{K}$ are mutually linear independent, where $\{\u_k\}_{k=0}^{K}$ are eigenvectors of $\cL$ spanning the \red{$K+1$-dimensional eigen-subspace} for eigenvalue $0$. 
\label{thm:closedForm}
\end{theorem}\noindent
See Appendix\;\ref{append:invertible} for a proof.

\vspace{0.05in}
\noindent
\textbf{Remarks}: 
\eqref{eq:MAP_sol_general} can be solved as a linear system $\bPhi \x^* = \H^\top \y$ efficiently without matrix inverse using \textit{conjugate gradient} (CG) \cite{axelsson1986rate}.  
For denoising $\H = \I$, and the conditions are trivially true.
For interpolation $\H$ is rectangular with a single $1$ in each row, and $\H\u_k \neq \0$ is satisfied if $\u_k$ has no zero entries.  
For deblurring $\H$ is typically full row-rank as a low-pass filter centered at different samples.
Since $\H$ is square and full row-rank, the nullspace has dimension $0$, \ie, $\H \u_k \neq \0, \forall k$. 
Thus, the key condition is the second one, which can be checked given known $\H$, $\u_0 = \1$ and $\u_k = \P_k, k \geq 1$. 

%% file: optPara.tex
Parameter $\mu$ in \eqref{eq:MAP_general} must be carefully chosen for best performance.
Analysis of $\mu$ for GLR-based signal denoising \cite{chen2017bias} showed that the best $\mu$ minimizes the \textit{mean square error} (MSE) by optimally trading off bias of the estimate with its variance. 
Here, we design a new method from a theorem in \cite{chen2017bias} to compute a near-optimal $\mu$ for \eqref{eq:MAP_general}. 

For simplicity, consider the denoising problem, where $\H = \I$, and hence $\bPhi = (\I + \mu \cL)$.  
Assume that observation $\y = \x^o + \z$ is corrupted by zero-mean iid noise $\z$ with covariance matrix $\bSigma = \sigma^2_z \I$.
Denote by $(\lambda_i,\v_i)$ the $i$-th eigen-pair of matrix $\cL$; the first $K+1$ eigen-pairs corresponding to eigenvalue $0$, $\lambda_1 = \ldots \lambda_{K+1} = 0$, are not important here.
We restate Theorem 2 in \cite{chen2017bias} as follows.
\begin{theorem}
MSE of solution \eqref{eq:MAP_sol_general} for $\H = \I$ is
\begin{align}
\text{MSE}(\mu) = \sum_{i=K+2}^{N} \psi_i^2(\v_i^{\top} \x^o)^2 + \sigma^2_z \sum_{i=1}^N \phi_i^2 
\label{eq:MSE}
\end{align}
where $\psi_i = \frac{1}{1 + \frac{1}{\mu \lambda_i}}$ and $\phi_i = \frac{1}{1 + \mu \lambda_i}$. 
\label{thm:MSE}
\end{theorem}
The two sums in \eqref{eq:MSE} correspond to bias square and variance of estimate $\x^*$, respectively.
In general, a larger $\mu$ entails a larger bias but a smaller variance; the best $\mu$ optimally trades off these two quantities for a minimum MSE. 

In \cite{chen2017bias}, the authors derived a corollary where $\text{MSE}(\mu)$ in \eqref{eq:MSE} is replaced by an upper bound function $\text{MSE}^+(\mu)$ that is strictly convex.
The optimal $\mu^*$ is then computed by minimizing convex $\text{MSE}^+(\mu)$ using conventional optimization methods. 
However, this upper bound is too loose in practice to be useful.
As an example, Fig.\;\ref{fig:mse} illustrates $\text{MSE}(\mu)$ for a 2D image \texttt{House} corrupted by iid Gaussian noise with standard deviation $10$, then GNG Laplacian $\cL$ is obtained as described previously.  
The computed $\mu^*=10^{-4}$ by minimizing the convex upper bound $\text{MSE}^+(\mu)$ is far from the true $\mu^* \approx 0.15$ by minimizing $\text{MSE}(\mu)$ directly (found by exhaustive search).



\begin{figure}
\centering
\begin{tabular}{cc}
\includegraphics[width=2.8in]{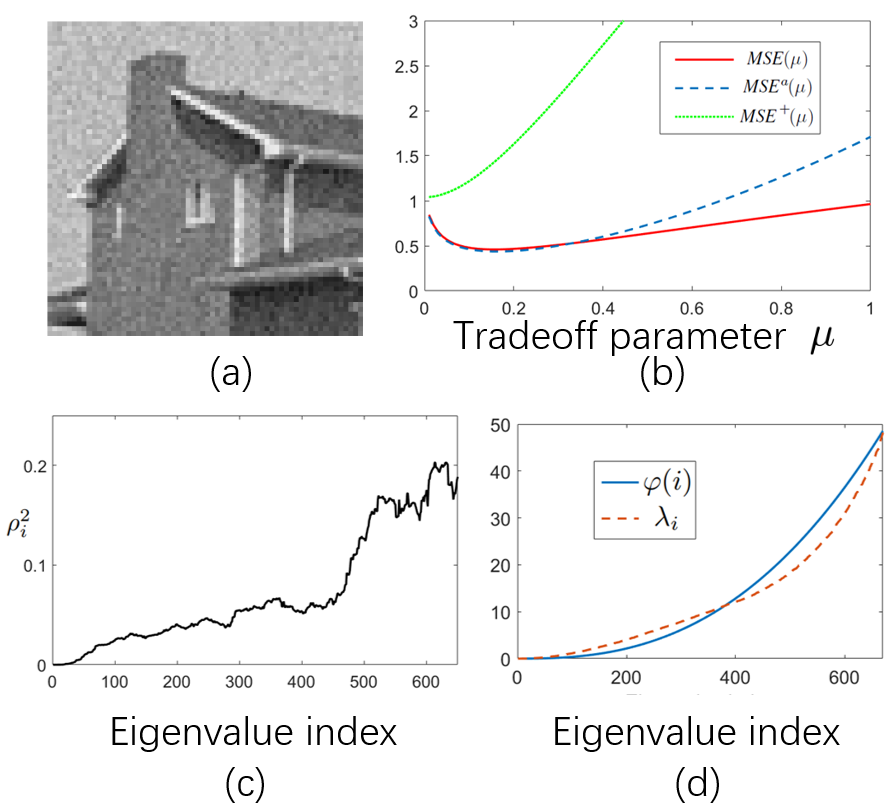}
\vspace{-0.1in}
\end{tabular}
\caption{\small (a) 2D image \texttt{House}. (b) $\text{MSE}(\mu)$, $\text{MSE}^+(\mu)$ and $\text{MSE}^a(\mu)$ as functions of weight parameter $\mu$ for \texttt{House}. 
(c) $\rho_i^2$ as function of eigenvalue index $i$.
(d) Modeling $\lambda_i$'s as exponential function $\varphi(i)$.}
\label{fig:mse}
\end{figure}

Instead, we take a different approach and approximate $\text{MSE}(\mu)$ in \eqref{eq:MSE} with a \textit{pseudo-convex} function \cite{Mangasarian_PseudoConvex1965}---one amenable to fast optimization. 
Ignoring the first $K+1$ constant terms in the second sum, we first rewrite \eqref{eq:MSE} as 
\begin{align}
MSE(\mu) \geq \sum_{i=K+2}^N \frac{\mu^2(\lambda_i\v_i^{\top}\x^o)^2+\sigma_z^2}{(1+\mu\lambda_i)^2}= \sum_{i=K+2}^N \frac{\mu^2 \rho_i^2+\sigma_z^2}{(1+\mu \lambda_i)^2}
\label{eq:MSE_1}
\end{align}
where $\rho_i= \lambda_i\v_i^{\top}\x^o$. 
\blue{We obtain an approximate upper bound by replacing each $\rho_i$ with $\max_i\rho_i = \rho_N \approx \lambda_N \v_N^\top \y$ in \eqref{eq:MSE_1}, assuming $\rho_i$ is monotonically increasing with $i$.} 
Extreme eigen-pair $(\lambda_N,\v_N)$ can be computed efficiently using LOBPCG \cite{lobpcg}.
See Fig.\;\ref{fig:mse}(c) for an illustration of $\rho_i^2$. 
Next, we model $\lambda_i$ as an exponential function of $i$,  $\varphi(i)=ai^b$, where $a$ and $b$ are parameters.
We thus approximate MSE as
\begin{align}
\text{MSE}^a(\mu)= \sum_{i=K+2}^N \frac{\mu^2 \rho_N^2+\sigma_z^2}{(1+\mu \varphi(i))^2} .
\label{eq:MSE_a}
\end{align}

In practice, we fit parameters of $\varphi(i)$, $a$ and $b$, using only extreme eigenvalues $\lambda_{K+2}$ and $\lambda_N$ (computable in linear time using LOBPCG \cite{lobpcg}) as $b=\frac{\ln{({\lambda_{K+2}}/{\lambda_N})}}{\ln{({K+2}/{N})}}$ and $a={(K+2)^b}/{\lambda_{K+2}}$.
We observe in Fig.\;\ref{fig:mse}(b) that $\text{MSE}^a(\mu)$ is a better upper bound than $\text{MSE}^+(\mu)$, leading to better approximate $\mu^* = 0.162$.
Since $\text{MSE}^a(\mu)$ is a differentiable and pseudo-convex function for $\mu>0$, \eqref{eq:MSE_a} can be minimized efficiently using off-the-shelf gradient-decent algorithms such as \textit{accelerated gradient descent} (AGD) \cite{Nesterov2013GradientMF}.

%% file: general.tex
We now consider a more general setting, where each node $i \in \cV$ in a manifold graph $\cG$ is not conveniently endowed with a coordinate \textit{a priori}.
To properly compute gradients, we present a parameter-free \textit{graph embedding} method that computes a vector $\p_i \in \mathbb{R}^K$ in a manifold space of dimension $K$, where $K \ll N$, for each node $i \in \cV$.
In essence, the computed graph embedding is a fast discrete approximation of the assumed underlying low-dimensional continuous manifold model.
Gradients can then be computed using obtained coordinates $\p_i$'s as done previously in Section\;\ref{sec:gglr}.

\subsection{Manifold Graphs} 
\label{subsec:detectManifold}

We assume that the input to our embedding is a \textit{manifold graph} defined in Section\;\ref{subsubsec:manifold}.
%
%
In the manifold learning literature \cite{carreira2005proximity, liu2011mixture,carey2017graph}, there exist many graph construction algorithms selecting node samples that closely approximate the hypothesized manifold. 
To evaluate quality of a constructed graph, \cite{carey2017graph} proposed several metrics; one example is \textit{betweenness centrality}, which measures how often a node $i$ appears in a shortest path between two nodes in the graph. 
Mathematically, it is defined as
\begin{align}
C_B(i)=\sum_{s,t\neq i} \frac{\sigma_{st}(i)}{\sigma_{st}}
\end{align}
where $\sigma_{st}$ is the number of shortest paths from nodes $s$ to $t$, and $\sigma_{st}(i)$ is the number of those paths that pass through node $i$.
Given a graph composed of nodes uniformly sampled from a smooth manifold, the betweenness centrality of nodes should be similar, \ie, all nodes are equally likely to appear in a given shortest path.
Thus, we first divide each $C_B(i)$ by the number of $(s,t)$ pairs, $(N-1)(N-2)$, and then employ the \textit{variance of betweenness centrality} (VBC) as a metric to evaluate the quality of a manifold graph. 
Only qualified manifold graphs are inputted to our algorithm to compute embeddings. 
As shown in Table\;\ref{tab:vbc}, the first four graphs with smaller VBCs are considered as qualified manifold graphs. 


\begin{table}   
\begin{center}   
\footnotesize
\caption{VBCs $(\times10^{-5})$ of graphs}
\label{tab:vbc} 
\begin{tabular}{|c|c|c|c|c|c|c|} 
\hline   AT\&T  & Football & FGNet & FIFA17 & Jaffe & AUS & Karate\\   
\hline    3.40  & 0.90 & 0.53 & 0.003  & 15.93  & 50.62 & 220\\ 
\hline 
\end{tabular}   
\end{center}
\vspace{-0.5cm}
\end{table}

\subsection{Defining Objective}
\label{subsec:embeddingObj}

Recall $\P \in \mathbb{R}^{N \times K}$ in the proof for Lemma\;\ref{lemma:eVec}, where the $i$-th row of $\P$ contains the $K$-dimensional vector $\p_i \in \mathbb{R}^K$ for node $i \in \cV$.
For notation convenience, we define here also $\q_k$ as the $k$-th column of $\P$---the $k$-th coordinate of all $N$ nodes. 
To minimize the distances between connected $1$-hop neighbors $(i,j) \in \cE$ in graph $\cG$, we minimize the GLR \cite{pang17}:
\begin{align}
\min_{\P \,|\, \P^{\top} \P = \I} \text{tr} \left( \P^{\top} \L \P \right) 
&= \sum_{k=1}^K \q_k^{\top} \L \q_k 
\label{eq:obj1} \\
&= \sum_{k=1}^K \sum_{(i,j) \in \cE} w_{i,j} (q_{k,i} - q_{k,j})^2 
\nonumber 
\end{align}
where $q_{k,i}$ is the $k$-th coordinate of node $i$.
Like LLE \cite{LLE}, condition $\P^{\top} \P = \I$ is imposed to ensure $\q_i^{\top} \q_j = \delta_{i-j}$. 
This orthgonality condition ensures $i$-th and $j$-th coordinates are sufficiently different and not duplicates of each other. 
Minimizing \eqref{eq:obj1} would minimize the squared Euclidean distance $\|\p_i - \p_j\|^2_2$ between connected node pair $(i,j)$ in the manifold space.
This objective thus preserves the \textit{first-order proximity} of the original graph structure \cite{xu21}. 

\begin{figure}
\centering
\begin{tabular}{cc}
\includegraphics[width=1.45in]{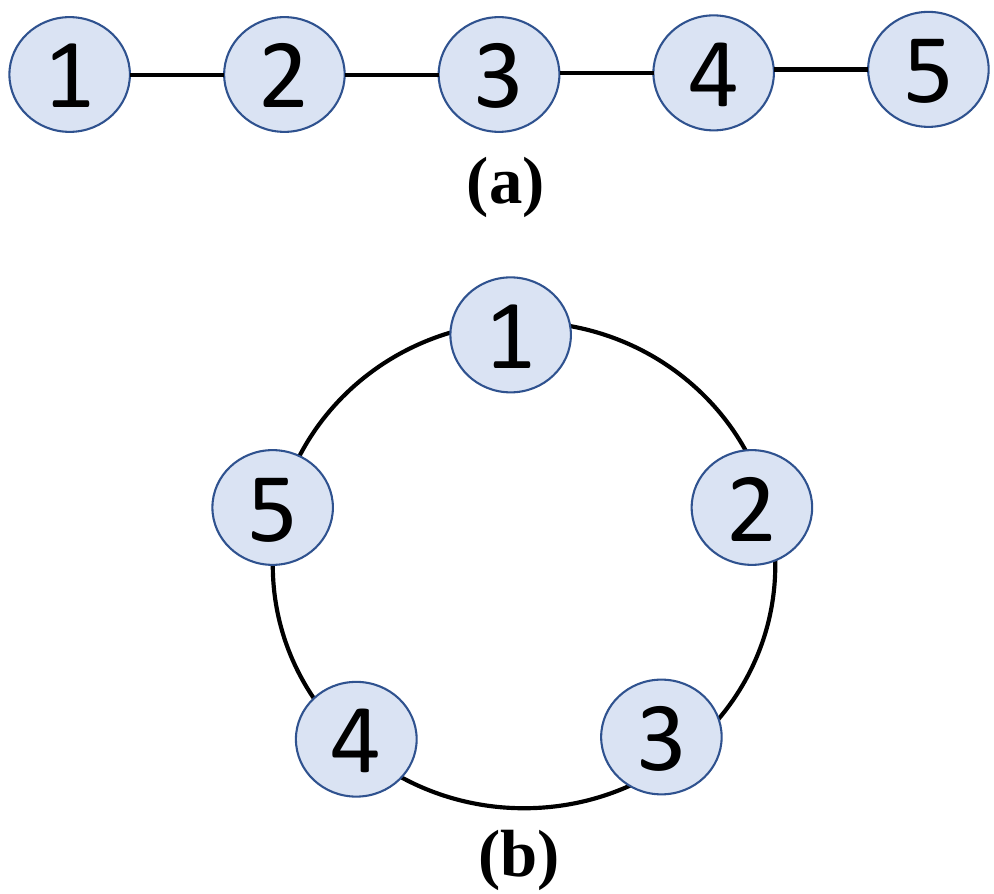}&\includegraphics[width=1.7in]{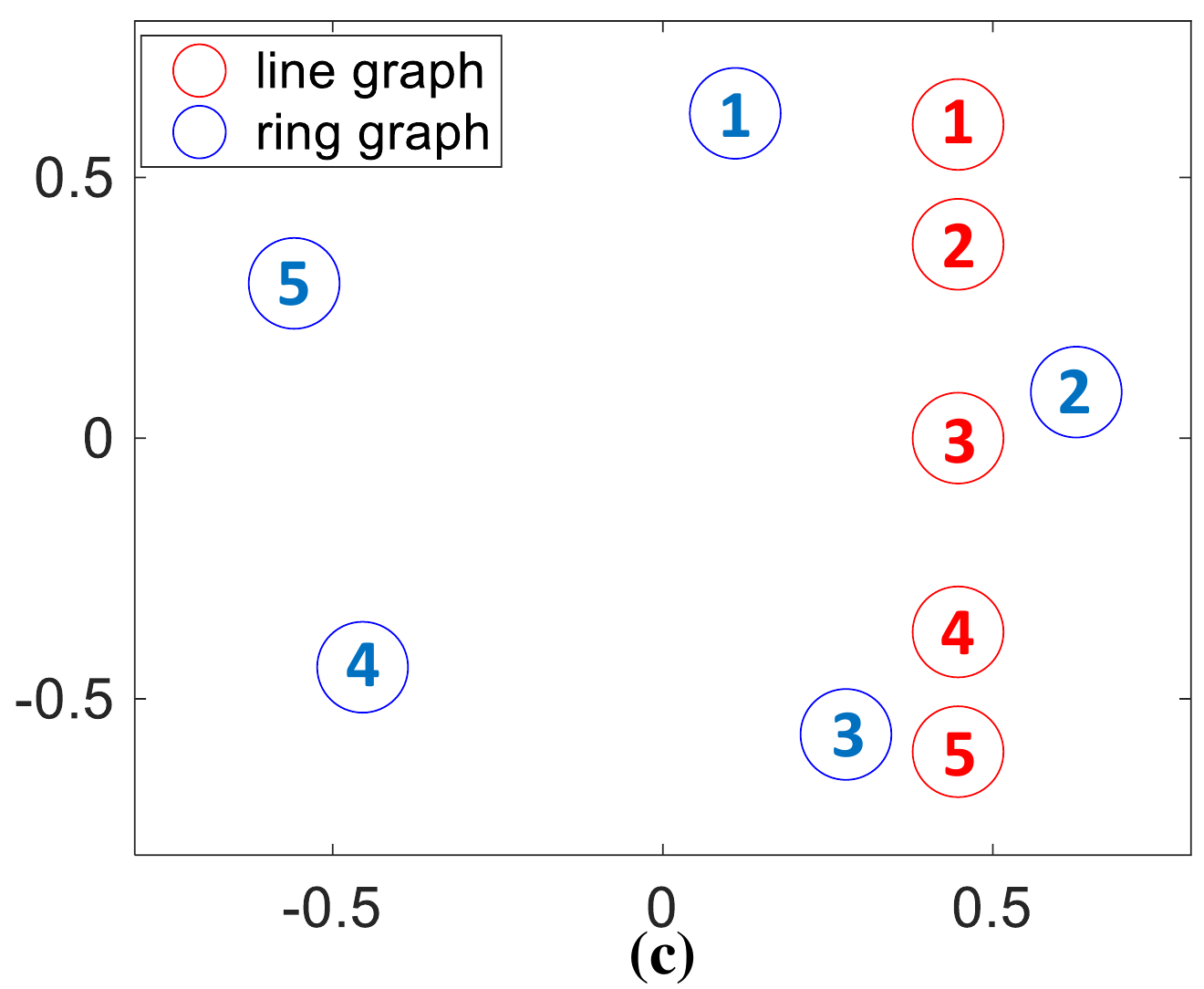}\\
\vspace{-0.25in}
\end{tabular}
\caption{\small Illustration of (a) a $5$-node line graph, (b) a $5$-node ring graph, where all nodes have the same degrees, and (c)  the first eigenvectors (\ie, 2D space vectors) of $\A$'s of (a) and (b).}
\label{fig:4node_graph}
\end{figure}

\subsubsection{$2$-hop Neighbor Regularization}

However, objective \eqref{eq:obj1} is not sufficient---it does not consider \textit{second-order proximity} of the original graph $\cG$. 
Consider the simple $5$-node line graph example in Fig.\;\ref{fig:4node_graph}(a). 
Just requiring each connected node pair to be in close proximity is not sufficient to uniquely result in a straight line solution (and thus in lowest dimensional manifold space). 
For example, a zigzag line in 2D space is also possible.

Thus, we regularize the objective \eqref{eq:obj1} using our second manifold graph assumption discussed in Section\;\ref{subsubsec:manifold} for $h=1$; 
specifically, if $(i,j) \in \cE$ but $(i,l) \not\in \cE$, then manifold distance $d_{i,j}$ between $(i,j)$ must be strictly smaller than distance $d_{i,l}$ between $(i,l)$, \ie, $d_{i,j} < d_{i,l}$. 

Based on this assumption, we define our regularizer $g(\P)$ as follows.
Denote by $\cT_i$ the \textit{two-hop neighbor} node set from node $i$; \ie, node $j \in \cT_i$ is reachable in two hops from $i$, but $(i,j) \not\in \cE$.
The aggregate distance between each node $i$ and its 2-hop neighbors in $\cT_i$ is
$\sum_{i \in \cV} \sum_{j \in \cT_i} \|\p_i - \p_j \|^2_2$.

We write this aggregate distance in matrix form.
For each $\cT_i$, we first define matrix $\bTheta_i \in \{0, 1\}^{N \times N}$ with entries

\vspace{-0.1in}
\begin{small}
\begin{align}
\Theta_{m,n} = \left\{ \begin{array}{ll}
\frac{1}{T_i} & \mbox{if}~ m = n = i ~~~\mbox{or}~~~ m = n \in \cT_i \\
-\frac{1}{T_i} & \mbox{if}~ m = i, n \in \cT_i ~~~\mbox{or}~~~ m \in \cT_i, n = i \\
0 & \mbox{o.w.}
\end{array} \right. 
\label{eq:Q}
\end{align}
\end{small}\noindent
where $T_i = |\cT_i|$ is the number of disconnected 2-hop neighbors.
We then define $\Q = \sum_{i \in \cV} \bTheta_i$. 
Finally, we define the regularizer as $g(\P) = - \gamma \, \text{tr}(\P^{\top} \Q \P) + \epsilon \I$.
Parameters $\gamma, \epsilon > 0$ are chosen to ensure matrix PSDness (to be discussed).
The objective becomes
\begin{align}
&~ \text{tr}(\P^{\top} \L \P) - \gamma \, \text{tr}(\P^{\top} \Q \P) + \epsilon \I 
\nonumber \\ 
&= \text{tr} (\P^{\top} \underbrace{\left( \L - \gamma \Q + \epsilon \I \right)}_{\A} \P ) .
\label{eq:obj2}
\end{align}
Note that objective \eqref{eq:obj2} remains quadratic in variable $\P$. 
Solution $\P^*$ that minimizes \eqref{eq:obj2} are first eigenvectors of $\A$, which can be computed in roughly linear time using LOBPCG \cite{lobpcg} assuming $K \ll N$.  
See Fig.\;\ref{fig:4node_graph}(c) for examples of computed first eigenvectors of $\A$ as graph embeddings for line and ring graphs shown in (a) and (b), respectively.


\subsection{Choosing Weight Parameter $\gamma$}
\label{subsec:weightPara}

As a quadratic minimization problem \eqref{eq:obj2}, it is desirable for $\A = \L - \gamma \Q + \epsilon \I$ to be PSD, so that the objective is lower-bounded, \ie,  $\q^{\top} \A \q \geq 0, \forall \q \in \mathbb{R}^N$. 
We set $\epsilon = \lambda^{(2)}_{\min}(\Q)$ to be the \textit{second} smallest eigenvalue---the Fiedler number---of $\Q$ (Laplacian has $\lambda^{(1)}_{\min}(\Q) = 0$); larger $\lambda^{(2)}_{\min}(\Q)$ means more disconnected 2-hop neighbors, and a larger $\gamma$ is desired. 
We compute $\gamma > 0$ so that $\A$ is guaranteed to be PSD via GCT \cite{varga04}.
Specifically, we compute $\gamma$ such that left-ends of all Gershgorin discs $i$ corresponding to rows of $\A$ (disc center $A_{i,i}$ minus radius $\sum_{j \neq i} |A_{i,j}|$) are at least 0, \ie,
\begin{align}
L_{i,i} - \gamma Q_{i,i} + \epsilon - \sum_{j | j \neq i} \left| L_{i,j} - \gamma Q_{i,j} \right| & \geq 0, ~~\forall i .
\label{eq:discLE1}
\end{align}

Note that $L_{i,j} = -W_{i,j} \leq 0$, and $Q_{i,j} \leq 0$.
Note further that node $j$ cannot both be a $1$-hop neighbor to $i$ and a disconnected $2$-hop neighbor at the same time, and hence either $L_{i,j}=0$ or $Q_{i,j}=0$. 
Thus, we remove the absolute value operator as
\begin{align}
L_{i,i} - \gamma Q_{i,i} + \epsilon - \sum_{j | j \neq i} \left( -L_{i,j} - \gamma Q_{i,j} \right) &\geq 0 .
\label{eq:discLE2}
\end{align}
We set the equation to equality and solve for $\gamma_i$ for row $i$, \ie, 

\vspace{-0.1in}
\begin{small}
\begin{align}
\gamma_i = \frac{L_{i,i} + \sum_{j|j\neq i} L_{i,j} + \epsilon}{Q_{ii} - \sum_{j|j\neq i} Q_{i,j}} 
= \frac{\epsilon}{Q_{ii} - \sum_{j|j\neq i} Q_{i,j}} ,
\label{eq:gamma}
\end{align}
\end{small}\noindent
where $L_{i,i} = - \sum_{j\neq i} L_{i,j}$.
Finally, we use the smallest non-negative $\gamma = \min_{i} \gamma_i$ for \eqref{eq:obj2} to ensure all disc left-ends are at least $0$, as required in \eqref{eq:discLE1}.

\begin{algorithm}[htb]
\caption{ Embedding for graph without coordinates}
\label{alg:alg_embedding}
\begin{small}
\begin{algorithmic}[1]
\REQUIRE~Graph $\cG$, manifold space dimension $K$. \\
\STATE Compute Laplacian $\L$ and 2-hop matrix $\Q$ via \eqref{eq:Q}. \\
\STATE Compute $\epsilon=\lambda_{\min}^{(2)}(\Q)$ using LOBPCG.\\ 
\STATE Compute $\gamma$ by solving \eqref{eq:gamma}.\\
\STATE Compute $\A=\L-\gamma\Q+\epsilon\I$.\\
\STATE Compute the $2$nd to the $K+1$-th eigenvector of $\A$ as $\P$.
\ENSURE~Embedding coordinates $\P$.  
\end{algorithmic}
\end{small}
\end{algorithm}

\begin{algorithm}[htb]
\caption{ SDGGLR for graph signal restoration}
\label{alg:alg_gglr}
\begin{small}
\begin{algorithmic}[1]
\REQUIRE~Observed signal $\y$, graph $\cG$, manifold space dimension $K$, observation matrix $\H$. \\
\STATE Initialize $\x=\y$.
\STATE If $\cG$ has no coordinates, compute embedding using Alg.\;\ref{alg:alg_embedding}.
\STATE Construct DAG $\cG^d$ from coordinates $\{\p_i\}$ and $\cG$. 
\STATE Construct gradient operators $\{\F^{i}\}$ from DAG $\cG^d$ via \eqref{eq:gradOp}.\\
\WHILE{not converge}
\STATE Compute manifold gradients $\balpha^i$ via \eqref{eq:alphai}. \\
\STATE Compute gradient graph adjacency matrix $\bar{\W}$ via \eqref{eq:gradWeight}.\\
\STATE Compute GNG Laplacian $\cL$ via \eqref{eq:GGLR}.\\
\STATE Update $\bPhi$, then compute $\x^*$ by solving \eqref{eq:MAP_sol_general}.
\ENDWHILE
\ENSURE ~~Restored signal $\x^*$.
\end{algorithmic}
\end{small}
\end{algorithm}




%% file: results.tex
We present a series of experiments to test our proposed GGLR as a regularizer for graph signal restoration. 
We considered four applications: 2D image interpolation, 3D point cloud color denoising, kNN graph based age estimation, and FIFA player rating estimation.
We first demonstrate SDGGLR's ability to reconstruct PWP signals in example image patches.

\subsection{PWC vs. PWP Signal Reconstruction}

\begin{figure}[t]
\begin{center}
   \includegraphics[width=0.99\linewidth]{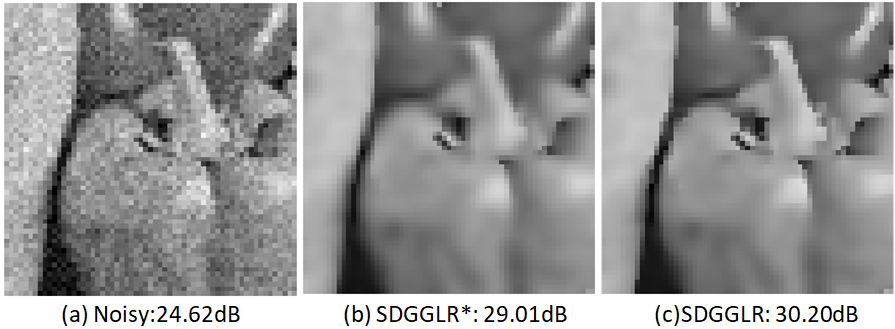}
\end{center}
\vspace{-0.2in}
\caption{
To denoise image in (a), SDGGLR without / with detecting and removing false gradients, promoting continuous / non-continuous PWP signal reconstruction, are shown in (b) and (c), respectively. 
(b) is visually more blurry than (c) with a noticeably worse reconstruction PSNR.
}
\label{fig:visualResult0}
\end{figure}
We constructed a 4-connected 2D grid graph for an image patch, where each pixel was a node, and edge weights were initially assigned $1$.
The collection of greyscale pixel values was the graph signal.
We first observe in Fig.\;\ref{fig:visualResult0} that, for image denoising, detection and removal of \textit{false gradients} (as discussed in Section\;\ref{subsec:PWP}) to promote non-continuous PWP signal reconstruction is important; continuous PWP signal reconstruction in (b) is noticeably more blurry than (c). 
We next visually compare reconstructed signals using SDGLR versus SDGGLR for regularization for image denoising and interpolation.  
Fig.\;\ref{fig:visualResult1}(a) shows a synthetic PWP image patch corrupted by white Gaussian noise of variance $6$.
Parameter $\sigma_{\alpha}$ in \eqref{eq:gradWeight} to compute gradient edge weights was set to $1.5$.
We estimated tradeoff parameter $\mu$ by minimizing \eqref{eq:MSE_a}. 
We solved \eqref{eq:MAP_general} iteratively---updating gradient edge weights via \eqref{eq:gradWeight} in each iteration---until convergence. 

We observe that SDGLR reconstructed an image with PWC (``staircase") characteristic in (b).
In contrast, SDGGLR reconstructed an image with PWP characteristic in (c), resulting in nearly 8dB higher PSNR.
Fig.\;\ref{fig:visualResult1}(d) is another synthetic PWP image patch corrupted by white Gaussian noise with variance $20$. 
We observe that SDGLR and SDGGLR reconstructed image patches with similar PWC and PWP characteristics, respectively.
Fig.\;\ref{fig:visualResult1}(g) shows a \texttt{Lena} image patch with $90\%$ missing pixels. 
Similar PWC and PWP characteristics in restored signals can be observed in (h) and (i), respectively.

\begin{figure}[t]
\begin{center}
   \includegraphics[width=0.99\linewidth]{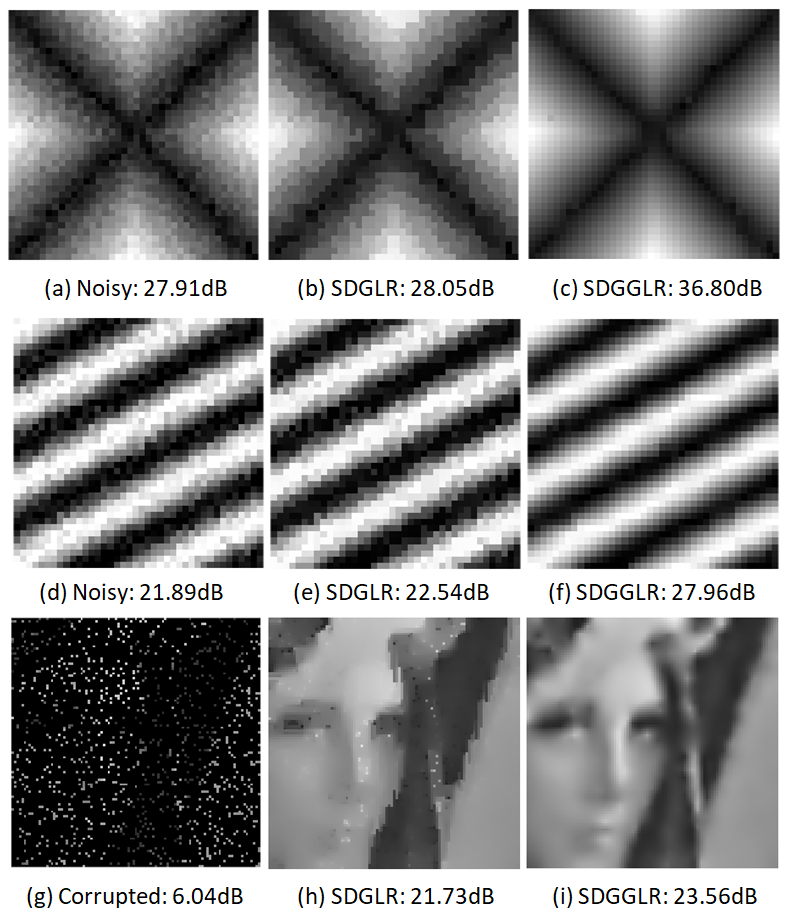}
\end{center}
\vspace{-0.2in}
\caption{Image patch denoising and interpolation results (PSNR) using SDGLR versus SDGGLR. Left column: Corrupted images. Middle column: Reconstructions using SDGLR for regularization. Right column: Reconstructions using SDGGLR for regularization.
}
\label{fig:visualResult1}
\end{figure}

\subsection{Performance on Graphs with Coordinates}
\subsubsection{Settings}

Image interpolation aims to estimate random missing grayscale pixel values in a 2D image. 
We constructed a 4-connected graph with initial edge weights set to $1$. 
Missing pixel values were set initially to $0$.
We can then compute gradients $\balpha^i$ via \eqref{eq:alphai} and gradient graph $\bar{\cG}$, where edge weights were updated via \eqref{eq:gradWeight}. 
We set tradeoff parameter $\mu = 0.01$ for the MAP formulation \eqref{eq:MAP_general}, which was solved iteratively till convergence. 

Three Middlebury depth images, \texttt{Cones}, \texttt{Teddy}, and \texttt{Dolls}\footnote{https://vision.middlebury.edu/stereo/data/}, were used.  
We compared SDGGLR against existing schemes: i) graph-based SDGLR\,\cite{pang17}, and ii) non-graph-based  TGV\,\cite{bredies2015tgv}, EPLL\,\cite{EPLLZoran2011}, IRCNN\,\cite{Zhang2017}, IDBP\,\cite{tirer19}, and GSC\,\cite{Zha2020}. 
For SDGLR, we set $\mu = 0.01$ for optimal performance. 
For TGV, EPLL, IDBP, and GSC, we set their respective parameters to default. 
IRCNN is a deep learning-based method, and we used the parameters trained by the authors.

For 3D point cloud color denoising (assuming that the point coordinates are noise-free), we conducted simulation experiments using datasets from  \cite{pointcloud_dataset}. 
We selected a voxel of $1200$ 3D points for testing and connected each point to its $20$ nearest neighbors in Euclidean distance to compose a kNN graph. 
Edge weights were computed using \eqref{eq:edgeWeight0}, where the feature and signal vectors were position and color vectors, respectively. 
Parameters $\sigma_f$ and $\sigma_x$ in \eqref{eq:edgeWeight0} and $\sigma_g$ in \eqref{eq:gradWeight} were set to $1$, $0.1$ and $10$, respectively. 
White Gaussian noise with standard deviation $\sigma_n=25, 50$ and $75$ was added to the luminance component. 
We compared our proposed SDGGLR against SDGLR\,\cite{pang17}, GTV\,\cite{Dinesh20193DPC}, and GSV\,\cite{chen15}. 
For fair comparison, GTV and GSV used the same kNN graph. 
The weight parameter for GTV and GSV was set to $0.5$. 

\subsubsection{Image Interpolation}

Table\;\ref{tab:numResult1} shows the resulting PSNR and \textit{structure similarity index measure} (SSIM) \cite{wang2004image} of the reconstructed images using different methods. 
The best results of each criterion are in boldface.
We first observe that SDGGLR was better than SDGLR in PSNR for all images and all missing pixel percentages, with a maxium gain of 2.45dB.
More generally, we observe that SDGGLR outperformed competing methods when the fraction of missing pixels was large. 
An example visual comparison is shown in Fig.\;\ref{fig:visualResult2}. SDGGLR preserved image contours and mitigated blocking effects, observable in reconstruction using GSC in (f).
SDGGLR achieved the highest PSNR and SSIM when the fraction of missing pixels was above $98\%$. 

\begin{figure}[t]
\begin{center}
   \includegraphics[width=0.99\linewidth]{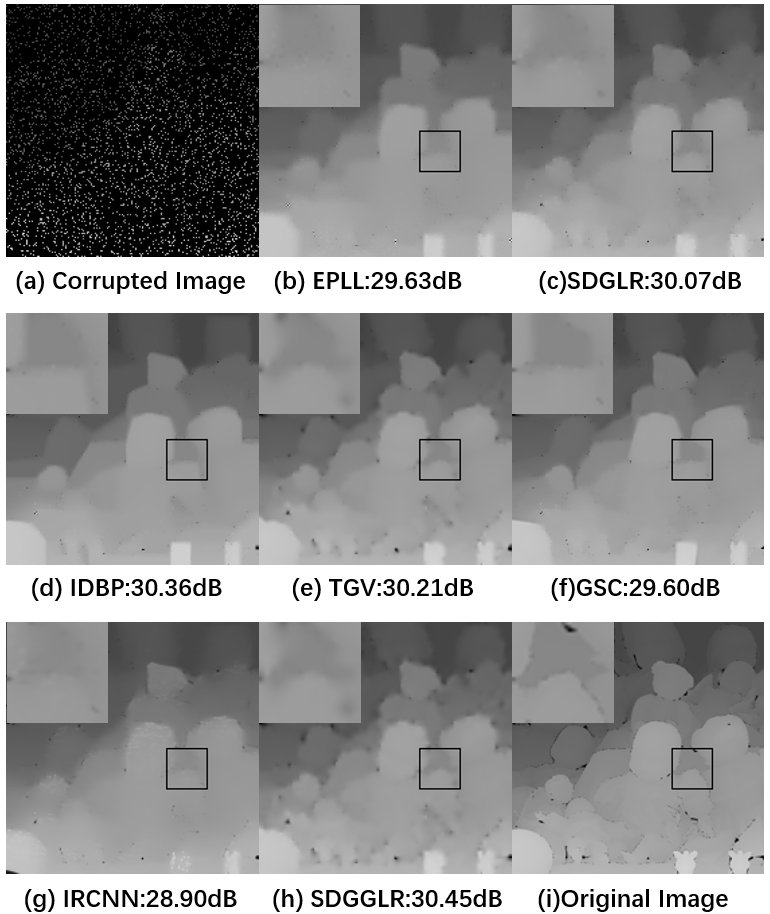}
\end{center}
\vspace{-0.2in}
\caption{Interpolation results (PSNR) using different methods on the corrupted image with $90\%$ missing pixels, selected at random. The resulting image using SDGGLR looks less blocky and more natural.} 
\label{fig:visualResult2}
\end{figure}

\begin{table}
\caption{\footnotesize Interpolation results by different methods on three depth images with different fraction of missing pixels.}
\label{tab:numResult1}
\vspace{-0.1in}
\tiny
\begin{center}
\begin{tabular}{|l|c|c|c|c|c|c|c|c|}
\hline
 \multicolumn{9}{|c|}{ 30\% missing pixels } \\ 
\hline
 Image & Metric&EPLL&SDGLR&IDBP &TGV& GSC&IRCNN&SDGGLR\\
\hline
 \multirow{2}{*} {Cones}& PSNR &  39.43  & 38.29 & 	39.93   &  40.09 & 	41.23  & \textbf{42.11}  & 40.74\\
\cline{2-9}  & SSIM   & 0.988   & 0.987 & 0.991   &  0.991  &\textbf{ 0.994} &\textbf{0.994}  &0.992 \\
\hline
 \multirow{2}{*} {Teddy}& PSNR & 38.89  & 37.84 & 38.35  & 38.26  & \textbf{40.06}  &  39.64 &38.35 \\
\cline{2-9}  & SSIM   &  0.994  & 0.990 &  0.994  & 0.994   & \textbf{0.996} &0.995  &0.994 \\
\hline
 \multirow{2}{*} {Dolls}& PSNR & 38.77  & 37.57 & 38.24  & 39.29  & 39.04  &  \textbf{40.07} & 39.36\\
\cline{2-9}  & SSIM   & 0.977   &0.979  & 0.979  & 0.983   & 0.982 & \textbf{0.984} & 0.983\\
\hline
 \multicolumn{9}{|c|}{ 60\% missing pixels } \\ 
\hline
 Image & Metric&EPLL&SDGLR&IDBP &TGV& GSC&IRCNN&SDGGLR\\
\hline
 \multirow{2}{*} {Cones}& PSNR & 34.87   & 33.41 & 35.70    & 35.14    & 	\textbf{36.19 }   &35.74  & 	35.59 \\
\cline{2-9}  & SSIM   & 0.971   & 0.968  & 0.977   & 0.974    & 0.980  & \textbf{0.981}  & 0.978 \\
\hline
  \multirow{2}{*} {Teddy}& PSNR & 32.78   & 31.80  &33.12    & 32.71   & \textbf{33.89}  & 	32.68   & 	33.19 \\
\cline{2-9}  & SSIM   &  0.978   & 0.972  &   0.982   &  0.979   & \textbf{0.986}  &  0.982 & 0.982 \\
\hline
 \multirow{2}{*} {Dolls}& PSNR & 33.20  & 33.43 & 	33.13  &34.32    &  33.52  &  34.02 & 	\textbf{34.35}\\
\cline{2-9}  & SSIM   & 0.932   & 0.944 & 0.938  & \textbf{0.949}  & 0.944 &  0.948 & \textbf{0.949}\\
\hline
 \multicolumn{9}{|c|}{ 90\% missing pixels } \\
\hline
 Image & Metric&EPLL&SDGLR&IDBP &TGV& GSC&IRCNN&SDGGLR\\
\hline
 \multirow{2}{*} {Cones}& PSNR &  27.98  &  	27.54&	29.31   & 	29.04  &  	29.05  &  	26.35& \textbf{29.78}\\
\cline{2-9}  & SSIM  &   0.888&	0.894&	0.917&	0.909&	\textbf{0.920}&	0.878&0.914 \\
\hline
 \multirow{2}{*} {Teddy}& PSNR &   26.28&	25.66&	26.92 &   	26.85&    	26.67&    	25.56 &   	\textbf{27.58} \\
\cline{2-9}  & SSIM   &   0.911&	0.910&	\textbf{0.928}&	0.921&	0.924&	0.909&	0.925 \\
\hline
 \multirow{2}{*} {Dolls}& PSNR &  29.63&    	30.07&	30.36 &   	30.21&    	29.60&	28.90 &   	\textbf{30.45} \\
\cline{2-9}  & SSIM   &    0.856&	0.870&	0.874&	0.876&	0.873&	0.857&	\textbf{0.876} \\
\hline
 \multicolumn{9}{|c|}{ 99\% missing pixels } \\ 
\hline
 Image & Metric&EPLL&SDGLR&IDBP &TGV& GSC&IRCNN&SDGGLR\\
\hline
 \multirow{2}{*} {Cones}& PSNR  &  10.94&    	24.04&	23.63&    	22.69 &  	23.02&    	7.64 & \textbf{24.57}
\\
\cline{2-9}  & SSIM   &  0.508&	0.801&	0.802&	0.775&	0.774&	0.118&	\textbf{0.803} \\
\hline
 \multirow{2}{*} {Teddy}& PSNR  &  13.09&    	22.086&	21.93&    	21.67  &  	21.49 &  	9.72&	\textbf{23.12} \\
\cline{2-9}  & SSIM    &  0.584&	0.824&	0.815&	0.805&	0.803&	0.234&	\textbf{0.830 }\\
\hline
 \multirow{2}{*} {Dolls}& PSNR &  11.83 &   	26.52&	26.42 &   	25.58 &   	25.09 &  	8.33 &  \textbf{26.62} \\
\cline{2-9}  & SSIM   & 0.509&	0.795&	0.791&	0.772&	0.787&	0.127&\textbf{0.796} \\
\hline
\end{tabular}
\end{center}
\end{table}

\begin{table}
\caption{\footnotesize Average runtime (in sec) on $128\times 128$ images.}
\label{tab:numResult2}
\vspace{-0.1in}
\footnotesize
\begin{center}
\begin{tabular}{|c|c|c|c|c|c|c|c|}
\hline
Methods&GSC&EPLL&IDBP&TGV &SDGGLR& SDGLR\\
\hline
Time (sec.) & $>60$  & 21.95 & 9.58& 9.27  & 2.53 & 1.81 \\
\hline
\end{tabular}
\end{center}
\end{table}

Table\;\ref{tab:numResult2} shows the runtime of different methods for a $128\times 128$ image. 
All experiments were run in the Matlab2015b environment on a laptop with Intel Core i5-8365U CPU of 1.60GHz.
TGV employed a primal-dual splitting method \cite{Condat2013} for $\ell_2$-$\ell_1$ norm minimization, which required a large number of iterations until convergence, especially when the fraction of missing pixels was large. 
In contrast, SDGGLR iteratively solved linear systems \eqref{eq:MAP_sol_general} and required roughly $2$sec. 
SDGLR's runtime was comparable, but its performance for image interpolation was in general noticeably worse.

\subsubsection{Point Cloud Color Denoising}

For point cloud color denoising, we computed the optimum tradeoff parameter $\mu$ for GGLR using methdology described in Section\;\ref{sec:optPara}. 
Fig.\;\ref{fig:psnr_mu} shows an example of 3D point cloud color denoising from dataset \texttt{Long Dinosaur}. 
We plotted the PSNR results versus chosen tradeoff parameter $\mu\in [0, 0.1]$. 
We see that there exists an optimum $\mu$ near $0.03$ for denoising. 
The computed $\mu$ using our proposed method \eqref{eq:MSE_a} is $0.0361$, which is quite close.
In contrast, minimizing a convex upper bound $\text{MSE}^+(\mu)$ as done in \cite{chen2017bias} resulted in $3.66\text{e}^{-6}$, which is far from optimal. 

\begin{figure}[t]
\begin{center}
   \includegraphics[width=0.5\linewidth]{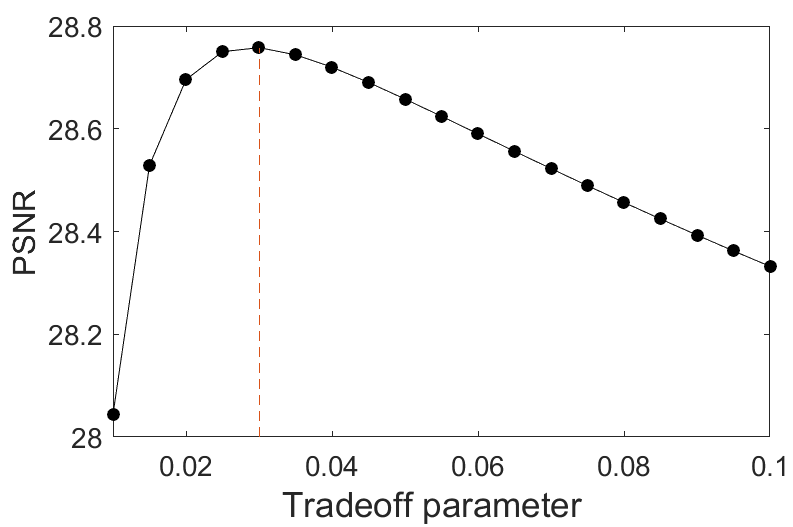}
\end{center}
\vspace{-0.2in}
\caption{Denoising results (PSNR) versus tradeoff parameter $\mu\in[0,0.1]$ for 3D point cloud \texttt{Long Dinosaur}.} 
\label{fig:psnr_mu}
\end{figure}

To quantitatively evaluate the effect of GGLR for denoising, Table\;\ref{tab:numResult_pc} shows the results in PSNR and PSSIM\;\cite{PSSIM2020}.
We see that PSNR of SDGGLR was better than SDGLR by roughly $0.5$dB on average. 
Visual results for point cloud \texttt{Long Dinosaur} are shown in Fig.\;\ref{fig:visualResult_pc1}. 
We see that GTV promoted PWS slightly better than SDGLR, resulting in fewer artifacts in the restored images. 
However, they over-smoothed local details of textures due to the PWC characteristic.
Since GGLR promoted PWP signal reconstruction, the restored point cloud color looks less blocky and more natural.

\begin{figure}[t]
\begin{center}
  \includegraphics[width=0.99\linewidth]{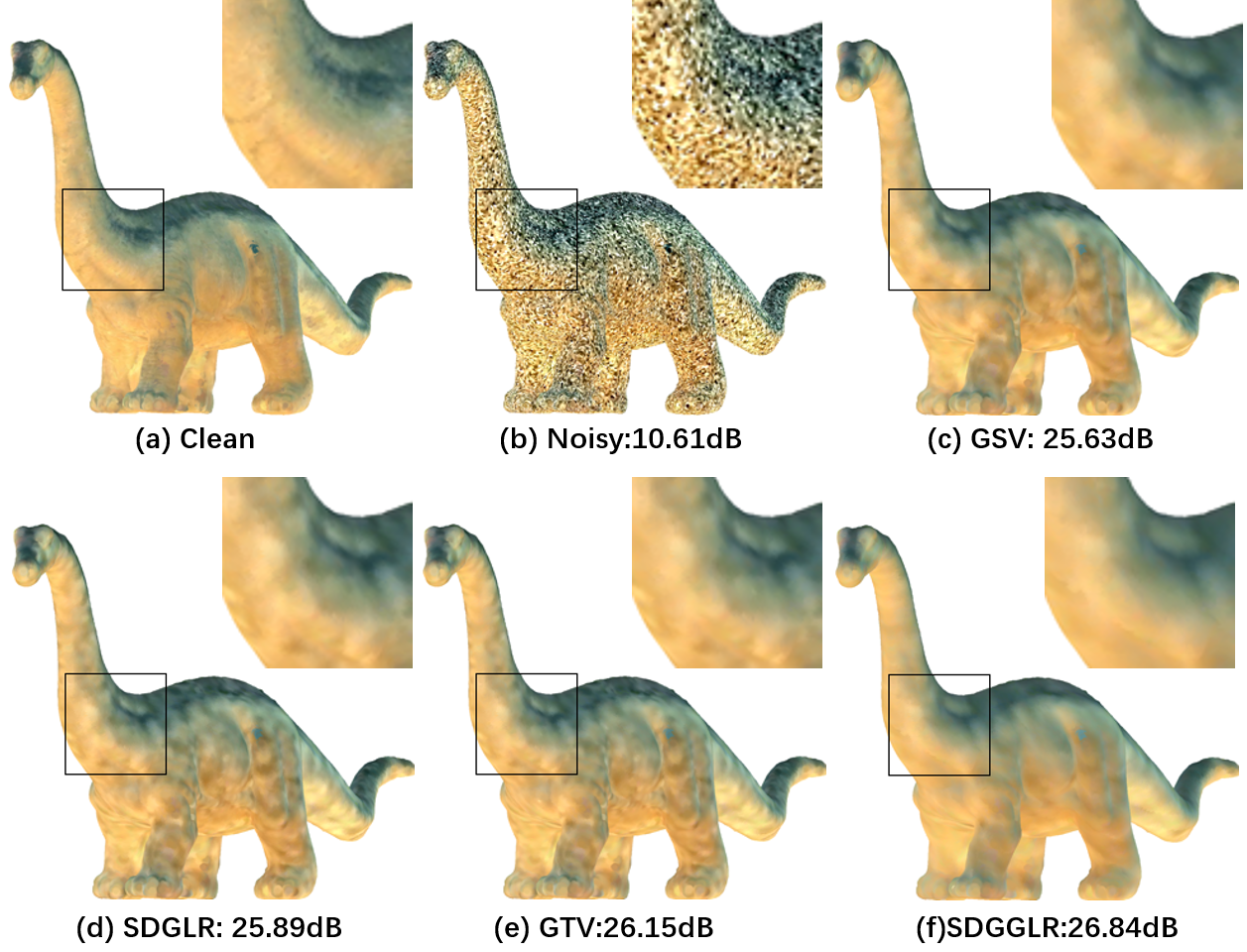}
\end{center}
\vspace{-0.2in}
\caption{Denoising results (PSNR) using different methods on the corrupted 3D point cloud \texttt{Long Dinosaur} with noise variance $\sigma_n=75$.} 
\label{fig:visualResult_pc1}
\end{figure}

  

\begin{table}
\caption{\footnotesize Color denoising  PSNR value(dB) for Gaussian noise with $\sigma_n=25, 50, 75$.}
\label{tab:numResult_pc}
\vspace{-0.1in}
\tiny
\begin{center}
\begin{tabular}{|l|c|c|c|c|c|c|c|c|c|}
\hline
 \multicolumn{9}{|c|}{ $\sigma_n=25$ } \\ 
\hline
Model & Metric & 4-arm &Asterix&Dragon&Green dino& Man&Statue&Long dino\\
 \hline
 noise &PSNR&20.15&	20.14&	20.16&	20.15&	20.09&	20.15&	20.15 \\
\hline
\multirow{2}{*} {GSV}&PSNR& 29.37&26.29&29.07&29.06&29.03 &28.91&29.22
 \\
      \cline{2-9}  &PSSIM &0.966&0.949&0.943&0.939&0.866& 0.953&0.967\\
\hline
\multirow{2}{*} {SDGLR}&PSNR&29.60&26.53&29.16&29.35& 29.73&29.01&29.85
   \\
\cline{2-9}  &PSSIM&0.967&0.946&0.946&0.946&0.894&0.953& 0.967\\
\hline
\multirow{2}{*} {GTV}&PSNR&\textbf{30.08}&27.30&29.20 & 29.75&\textbf{29.90}&29.06&\textbf{30.22}
   \\
   \cline{2-9}  & PSSIM&\textbf{0.968}&0.945&0.944&0.945 &0.893&0.954&\textbf{0.969}\\
\hline
\multirow{2}{*} {SDGGLR}&PSNR&29.70&\textbf{28.62}&\textbf{29.57} &\textbf{29.95}&29.77&\textbf{29.09}&30.17
\\
   \cline{2-9}&PSSIM&0.967&\textbf{0.960}&\textbf{0.947}&\textbf{0.946}&\textbf{0.920}&\textbf{0.955}&\textbf{0.969}  \\
\hline
 \multicolumn{9}{|c|}{ $\sigma_n=50$ } \\ 
\hline
Model & Metric & 4-arm &Asterix&Dragon&Green dino& Man&Statue&Long dino\\
 \hline
 noise &PSNR&14.13&	14.12&14.14&14.13&14.07&14.13&14.13\\
\hline
\multirow{2}{*} {GSV}&PSNR&26.65&25.03&26.26&26.48&26.39& 26.15&26.52 
\\
   \cline{2-9}  &PSSIM&0.944&0.930&0.907&0.902&0.790& 0.924&0.945\\
\hline
\multirow{2}{*} {SDGLR}&PSNR&27.12&25.18&26.57&26.80&  \textbf{27.41}&26.36&27.08
  \\
     \cline{2-9}  &PSSIM&0.948&0.930&0.917&0.908&0.830&0.928&0.952\\
\hline
\multirow{2}{*} {GTV}&PSNR& \textbf{27.30}  & 25.36 &  \textbf{26.89} &  27.10 & 27.03 & \textbf{26.63} &  27.36
  \\
     \cline{2-9}  & PSSIM &\textbf{0.950} & 0.934 & \textbf{0.918} &0.913 & 0.831  &\textbf{0.933} &0.954\\
\hline
\multirow{2}{*} {SDGGLR}&PSNR& 27.05 &  \textbf{25.74} &  26.82 & \textbf{27.42} &  26.97 & 26.42 &  \textbf{27.97} \\
   \cline{2-9}  & PSSIM&\textbf{0.950} & \textbf{0.935} & \textbf{0.918} & \textbf{0.915}& \textbf{0.850} & 0.933 & \textbf{0.955}  \\
\hline
 \multicolumn{9}{|c|}{ $\sigma_n=75$ } \\ 
\hline
  Model & Metric& 4-arm &Asterix&Dragon&Green dino& Man&Statue&Long dino\\
 \hline
 noise &PSNR&10.61&
	10.60&10.61&	10.61&	10.55&	10.61&	10.61
 \\
\hline
\multirow{2}{*} {GSV}&PSNR& 25.25 & 22.94 & 24.79 & 25.13 & 24.91 & 24.66 & 25.63
   \\
      \cline{2-9}  & PSSIM &0.933 &   0.912  &  0.887 &   0.883   & 0.753  &  0.908  &  0.940\\
\hline
\multirow{2}{*} {SDGLR}&PSNR&26.03  & 23.31 & 25.29 &  25.47 & \textbf{25.83} & 24.80 &  25.89
 \\
    \cline{2-9}  & PSSIM&0.935  &0.914 & 0.896 & 0.885 &  0.777 &  0.911  & 0.940 \\
\hline
\multirow{2}{*} {GTV}&PSNR&   \textbf{26.34} & 23.58 &  25.38 &  25.96 &  25.61  & 24.88  & 26.15
  \\
     \cline{2-9}  & PSSIM & 0.938   & 0.915  &0.897&0.893 & 0.772& 0.915 &  0.939\\
\hline
\multirow{2}{*} {SDGGLR}&PSNR& 26.04 &  \textbf{23.69} & \textbf{25.62} &  \textbf{26.21} & 25.54 & \textbf{25.32} &  \textbf{26.84}
 \\
    \cline{2-9}  &PSSIM &\textbf{0.940} & \textbf{0.915} &  \textbf{0.906} &\textbf{0.896} & \textbf{0.803} & \textbf{0.920} & \textbf{0.945}\\
\hline
\end{tabular}
\end{center}
\end{table}

\subsection{Performance on Graphs without Coordinates}

\subsubsection{Settings}
To evaluate GGLR for graphs without coordinates, we consider two datasets: \texttt{FGNet} for age estimation and \texttt{FIFA17} for player-rating estimation. 
\texttt{FGNet} is composed of a total of $1002$ images of $82$ people with age range from $0$ to $69$ and an age gap up to $45$ years. 
We constructed kNN graph ($k=30$) for \texttt{FGNet} based on Euclidean distances between facial image features. 
Here, face images were nodes, and the graph signal was the age assigned to the faces. 
The VBC of kNN graph is generally very small and qualified as a manifold graph; in this case, VBC is $5.3e^{-6}$. 
Here, we considered an age interpolation problem: 
some nodes had missing age information, and the task was to estimate them. 
Specifically, we randomly removed age information from $90\%$, $70\%$, $40\%$, and $10\%$ of nodes---we call this the \textit{signal missing ratio}.
We first conducted experiments to test our embedding method, and then evaluated our proposed GGLR on the embedded latent space. 

\texttt{FIFA17} includes the statistics for all players according to FIFA records. 
It was downloaded from the Kaggle website. 
We collected club, age, and rating information from $1618$ England players for our experiments.
The players were represented by a graph where nodes represented players, and edges connected players with the same club or age. 
The graph has a total of $94092$ edges.
VBC for this graph is  $3e^{-8}$---small enough to qualify as a manifold graph.
The graph signal here was the player rating. 
Similarly, we computed a graph embedding and then evaluated GGLR on the latent space. 

For manifold graph embedding, we computed eigenvector matrix $\P$ for matrix $\A$ using Algorithm\,1, and checked how many eigenvectors were required to significantly reduce a normalized variant of cost function $\text{tr}(\P^\top \A \P)$ in \eqref{eq:obj2}.
For both \texttt{FGNet} and \texttt{FIFA17}, the latent spaces were chosen as two dimensions.

\subsubsection{Age Estimation}

We applied our proposed embedding method and GGLR to age estimation. 
We first compared our embedding using general eigenvectors (GE) against LLE \cite{LLE} and LE \cite{LE}. 
From Table\;\ref{tab:numResult_age_rating}, we observe that GGLR on the coordinate vector provided by LLE and LE were not stable. 
They could not obtain better interpolation results, compared to the results obtained using GE. 
We also compared SDGGLR against SDGLR and GTVM. 
\blue{ Quantitative results in peak signal-to-noise ratio (PSNR) and accuracy (Acc) are shown in Table\;\ref{tab:numResult_age_rating}. PSNR is a measure computed from mean squared error; a higher PSNR means a higher signal quality.
Accuracy is the number of correct predictions divided by the total number of predictions. }
We observe that PSNR and accuracy of SDGGLR were better than GTVM, and slightly better than SDGLR when the missing ratio was smaller than 70\%. 

\begin{table}
\caption{\footnotesize PSNR value(dB) and accuracy by different methods on dataset \texttt{FGNet} and \texttt{FIFA17} with different missing ratio.}
\label{tab:numResult_age_rating}
\vspace{-0.1in}
\tiny
\begin{center}
\begin{tabular}{|l|c|c|c|c|c|c|c|c|c|}
\hline
\multirow{3}{*} {Dataset}&\multirow{3}{*} {Method}& \multicolumn{8}{|c|}{Signal missing ratio} \\
\cline{3-10}
   & & \multicolumn{2}{|c|}{90\%}&\multicolumn{2}{|c|}{70\%}&\multicolumn{2}{|c|}{40\%}&\multicolumn{2}{|c|}{10\%}\\
\cline{3-10}
   & & PSNR&\blue{Acc}&PSNR&\blue{Acc}&PSNR&\blue{Acc}&PSNR&\blue{Acc}\\
\hline
 \multirow{5}{*} {FGNet}&SDGLR& \textbf{33.71} & \blue{\textbf{0.480}} &  37.82  & \blue{0.752} & 39.42 & \blue{0.898} & 44.46 & \blue{0.937}\\
\cline{2-10}  & GTVM&32.57 &   \blue{0.457 }& 37.56 &  \blue{0.670} & 39.33  &  0.833 & 43.67 &  \blue{0.923}\\
\cline{2-10}  & LE+GGLR&32.21& \blue{0.471} &  37.80 &\blue{0.756} & 40.48& \blue{0.904} &  45.40 &\blue{0.964}\\
\cline{2-10}  & LLE+GGLR&32.78 & \blue{0.467}& 38.38 &\blue{0.759} &    40.88& \blue{0.906} &   45.15& \blue{0.968}\\
\cline{2-10}  & GE+GGLR& 33.63& \blue{0.476} &  \textbf{39.19}& \blue{\textbf{0.761}} & \textbf{43.81}& \blue{\textbf{0.920}} &  \textbf{48.30} &\blue{\textbf{0.976}}\\
\hline
 \multirow{5}{*} {FIFA17}&SDGLR&25.65&\blue{0.167}&28.71&\blue{0.468}&33.48&\blue{0.748}&38.15&\blue{0.915} \\
\cline{2-10}  & GTVM&24.58&\blue{0.159}&28.65&\blue{0.453}&32.96&\blue{0.721}&37.58&\blue{0.908}\\
\cline{2-10}  & LE+GGLR&25.51 & \blue{0.164} &  \textbf{29.85}& \blue{\textbf{0.484}}&  33.73 &\blue{0.753} & 38.54& \blue{0.917}\\
\cline{2-10}  & LLE+GGLR&24.64 & \blue{0.171}& 28.93& \blue{0.478} &  33.25 &\blue{0.750}  &  38.63& \blue{0.921}\\
\cline{2-10}  & GE+GGLR&\textbf{25.74}& \blue{\textbf{0.172}} &  29.57& \blue{0.480} &   \textbf{34.20}& \blue{\textbf{0.767}} &   \textbf{39.01}& \blue{\textbf{0.923}} \\
\hline
\end{tabular}
\end{center}
\end{table}

\subsubsection{Player Rating Estimation}

From Table\;\ref{tab:numResult_age_rating}, we observe that GGLR on three embedding methods achieved better interpolation results compared to the results using GLR and GTVM. 
We also compared SDGGLR against SDGLR and GTVM. 
Quantitative results in PSNR are shown in Table\;\ref{tab:numResult_age_rating}. 
We observe that PSNR and accuracy of SDGGLR were better than SDGLR by nearly $1$dB on average.

%% file: conclude.tex
Unlike graph Laplacian regularizer (GLR) that promotes piecewise constant (PWC) signal reconstruction, we propose gradient graph Laplacian regularizer (GGLR) that promotes piecewise planar (PWP) signal reconstruction. 
For a graph signal endowed with sampling coordinates, we construct a gradient graph on which to define GLR, which translates to gradient-induced nodal graph (GNG) Laplacian in the nodal domain for regularization.
For a signal on a manifold graph without sampling coordinates, we propose a fast parameter-free method to first compute manifold coordinates. 
Experimental results show that GGLR outperformed previous graph signal priors like GLR and graph total variation (GTV) in a range of graph signal restoration tasks.

%% file: append.tex
\ifnum\arXiv=1
    
\subsection{\blue{Computing Graph Gradients using Normalized Laplacian}}
\label{append:normL}

\blue{
Denote by $\L$ and $\L_n$ the combinatorial graph Laplacian and symmetric normalized graph Laplacian matrices of a positive connected graph $\cG$, respectively. 
We prove that the only signal $\x$ such that $(\L_n \x)^\top \L (\L_n \x) = 0$ is the first eigenvector $\v_1$ of $\L_n$ corresponding to smallest eigenvalue $0$. 
$\L_n$ is defined as $\L_n \triangleq \D^{-1/2} \L \D^{-1/2} = \I - \D^{-1/2} \W \D^{-1/2}$, where $\D$ and $\W$ are the degree and adjacency matrices, respectively, and $\I$ is the identity matrix. 
Eigenvector $\v_1$ of $\L_n$ corresponding to the smallest eigenvalue $0$ is $\d^{1/2}$, where $d_i = \sum_j W_{i,j}$ is the degree of node $i$, \ie,
\begin{align}
\L_n \v_1 &= \D^{-1/2} \L \D^{-1/2} \d^{1/2} \\
&= \D^{-1/2} \L \1 = \0 .
\end{align}
}

\blue{
Since $\L$ is PSD with (unnormalized) eigenvector $\u_1 = \1$ corresponding to smallest eigenvalue $0$, for $\y^\top \L \y = 0$, $\y$ must be a constant vector, \ie, $\y = c \1$ for some constant $c \neq 0$. 
In general, any vector $\x = \sum_{n=1}^N a_n \v_n$ can be expressed as a linear combination of eigenvectors $\{\v_n\}$ of $\L_n$ that are orthogonal basis vectors spanning the signal space $\cH = \mathbb{R}^{N}$, where $a_i = \langle \x, \v_i \rangle \triangleq \v_i^\top \x$. 
We now write
\begin{align}
\L_n \x &\stackrel{(a)}{=} \left( \sum_{i=2}^N \lambda_i \v_i \v_i^\top \right) \left(\sum_{n=1}^N a_n \v_n \right) \\
&\stackrel{(b)}{=} \sum_{i=2}^N \lambda_i a_i \v_i
\end{align}
where in $(a)$ we write $\L_n$ as a linear combination of rank-1 matrices that are outer-products $\v_i \v_i^\top$ of its eigenvectors $\v_i$, each scaled by eigenvalue $\lambda_i$. 
$(b)$ follows since eigenvectors of symmetric $\L_n$ are orthonormal, \ie, $\v_i^\top \v_j = \delta_{i-j}$. 
Note that there is no $\v_1 \v_1^\top$ in the summation, since $\lambda_1 = 0$. 
Thus, in order for $\y^\top \L \y = 0$, $\y$ must equal $c\1 = \L_n \x = \sum_{i=2}^N \lambda_i a_i \v_i$.
However, the inner-product $a_1 = \langle c\1, \v_1 \rangle = \v_1^\top c \1 = c \sum_i v_{1,i} = c \sum_i d_i^{1/2}$, where $d_i > 0, \forall i$ for positive connected $\cG$. 
Thus, $\sum_i d_i^{1/2} > 0$, and $a_1 = c \sum_i d_i^{1/2} \neq 0$ for $c \neq 0$. 
Thus, $\y$ cannot be written as $c \1 = \sum_{i=2}^N \lambda_i a_i \v_i$, except for the case when $c=0$ and $a_i = 0, \forall i$. 
We can thus conclude that the only signal where $(\L_x \x)^\top \L (\L_n \x) = 0$ is the first eigenvector $\v_1$ of $\L_n$ corresponding to eigenvalue $0$. 
}

\fi

\subsection{Proof of Lemma\;\ref{lemma:acyclic}}
\label{append:acyclic}

\begin{proof}
Denote each constructed directed edge by its \textit{displacement vector}, $(0, \ldots, 0, p_{j,k^o} - p_{i,k^o}, \ldots, )$, \ie, a sequence of $k^o-1$ zeros, followed by a \textit{positive} difference between node pair $[i,j]$ at coordinate $k^o$, followed by differences in the remaining coordinates (of any signs). 
A path $\cP$ is a sequence of connected edges, and it is a cycle iff the sum of corresponding displacement vectors is the zero vector. 
Denote by $l$ the first non-zero (positive) coordinate of all edges in $\cP$. 
Because coordinate $l$ of all edges has entries $\geq 0$, the sum of all entries at coordinate $l$ for $\cP$ must be strictly greater than $0$.
Hence $\cP$ cannot be a cycle.
\end{proof}

\subsection{Proof of Lemma\;\ref{lemma:localGradient}}
\label{append:locality}

We prove Lemma\;\ref{lemma:localGradient}---nodes selected to compute manifold gradient $\balpha^i$ for node $i$ are the closest $K^+$ nodes satisfying the acyclic condition.
The second property of manifold graph (Section\;\ref{subsubsec:manifold}) implies $d(i,j) < d(i,l)$ if nodes $j$ and $l$ are $h$- and $h+1$-hop neighbors of $i$ and $\exists (j,l) \in \cE$. 
Thus, prioritizing $h$-hop neighbor $j$ into candidate list $\cS$ before $h+1$-hop neighbor $l$, where $\exists (j,l) \in \cE$, means $\cS$ always contains the closer node $j$ to $i$ than $l$.
Since $h+1$-hop neighbors $l$ are next added to $\cS$ if $h$-hop neighbor $j$ is closer to $i$ than other nodes in $\cS$, $\cS$ always contains the closest node not already chosen.
Hence, by always selecting the closest node $j^*$ in list $\cS$ in each iteration, the procedure chooses the $K^+$ closest nodes to $i$ satisfying the acyclic condition.
$\Box$

\subsection{Proof of Theorem\;\ref{thm:GGLR}}
\label{append:GGLR}

\begin{proof}
By definition, $\cL = \G^\top \R^\top \bar{\L}^o \R \G$, where $\bar{\L}^o$ is a graph Laplacian matrix corresponding to a positive graph, which is provably PSD \cite{cheung18}.  
Hence, we can eigen-decompose $\bar{\L}^o$ into $\bar{\L}^o = \U \bSigma \U^\top$, where $\bSigma = \text{diag}(\{\lambda_i\})$ is a diagonal matrix of non-negative eigenvalues $\lambda_i \geq 0, \forall i$. 
\blue{Then, to show $\cL$ is also PSD is straightfoward \cite{golub12}}:
\begin{align*}
\x^\top \cL \x  &= \x^\top \G^\top \R^\top \U \text{diag}(\{\lambda_i\}) \U^{\top} \R \G \x \\
&\stackrel{(a)}{=} \y^\top \text{diag}(\{\lambda_i\}) \y  
= \sum_i \lambda_i y_i^2 \geq 0
\end{align*}
where in $(a)$ we define $\y \triangleq \U^\top \R \G \x$.
Since this is true $\forall \x$, $\cL$ is PSD, meaning $\x^\top \cL \x \geq 0, \forall \x$.

\blue{
Suppose signal $\tilde{\x}$ is a planar signal defined in \eqref{eq:planar} and only contains discrete points exactly on a hyperplane parameterized by $\balpha^*$. 
Then, using \eqref{eq:alphai}, any $\balpha^i$ evaluates to $\balpha^*$ regardless of $\W^i$, since $\C^i \balpha^* = \F^i \tilde{\x}$ minimizes the square error objective in \eqref{eq:WLS}.  
Thus, $\tilde{\balpha} = [\balpha^*; \ldots; \balpha^*]$ computed using \eqref{eq:alpha_vec} contains the \textit{same} gradient $\balpha^*$ for all $i \in \bar{\cV}$, and $\text{Pr}(\tilde{\balpha}) = \tilde{\balpha}^\top \bar{\L}^o \tilde{\balpha} = \tilde{\x}^\top \cL \tilde{\x} = 0$. 
Since $\tilde{\x}$ achieves the minimum value for GGLR $\x^\top \cL \x \geq 0$, GGLR promotes planar signal reconstruction.
}
\end{proof}

\subsection{Proof of Theorem\;\ref{thm:closedForm}}
\label{append:invertible}

\red{
To prove Theorem\;\ref{thm:closedForm}, we prove that $\bPhi \triangleq \H^\top \H + \mu \cL$ is invertible under the said conditions. 
First, $\H^\top \H$ and $\cL$ are PSD (Theorem \ref{thm:GGLR}), and $\mu > 0$; thus, $\bPhi$ is PSD by Weyl's Inequality \cite{Nathanson1996Weyls}.  
By Lemma\;\ref{lemma:eVec} and \ref{lemma:eVal}, $\cL$ has $K+1$ (non-orthogonal but linear independent) eigenvectors $\{\u_k\}_{k=0}^K$ for eigenvalue $0$.
Define $\s \triangleq \sum_k a_k \u_k$, where $a_k \in \mathbb{R}$; we know $\s^\top \cL \s = 0$. 
We show next that, under the said conditions, $\s^\top \H^\top \H \s = \|\H \s\|_2^2 > 0$, or equivalently, $\H \s \neq \0$, if $\exists k, a_k \neq 0$.
Thus, there are no $\s \neq \0$ such that $\s^\top \H^\top \H \s$ and $\s^\top \cL \s$ are both zero, and thus $\bPhi$ is PD and invertible.
}

\red{
We first rewrite $\H\s$ as
\vspace{-0.05in}
\begin{align}
\H \s = \sum_{k = 0}^K a_k \H \u_k .
\end{align}
By the first condition, $\H \u_k \neq \0, \forall k$. 
By the second condition, $\{\H\u_k\}$ are mutually linearly independent, and thus there do not exist $\{a_k\}$ such that $\sum_k a_k \H \u_k = \0$ except $a_k = 0, \forall k$.
Thus, $\H \s \neq \0$ or $\|\H\s\|^2_2 > 0$ if $\exists k, a_k \neq 0$, and we conclude that $\bPhi$ is PD and invertible.
$\Box$
}

%% file: bio.tex
\begin{IEEEbiography}
[{\includegraphics[width=1in,height=1.25in,clip,keepaspectratio]{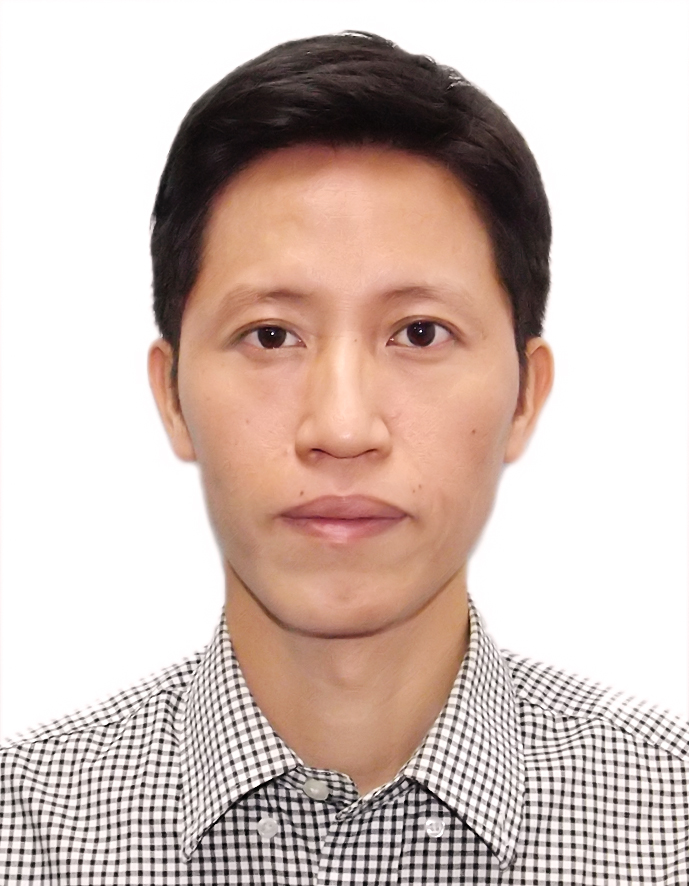}}]{Fei Chen} (S'09--M'12)  received the Ph.D. degree in signal and information processing from Zhejiang University, Hangzhou, China, in 2013. He is currently an Associate Professor with the College of Computer and Data Science, Fuzhou University, Fuzhou, China. His current research interests include machine learning, computer vision, and graph signal processing. He has served as associate editor for the Journal of Algorithms and Computational Technology (2022-present).
\end{IEEEbiography}

\begin{IEEEbiography}[{\includegraphics[width=1in,height=1.25in]{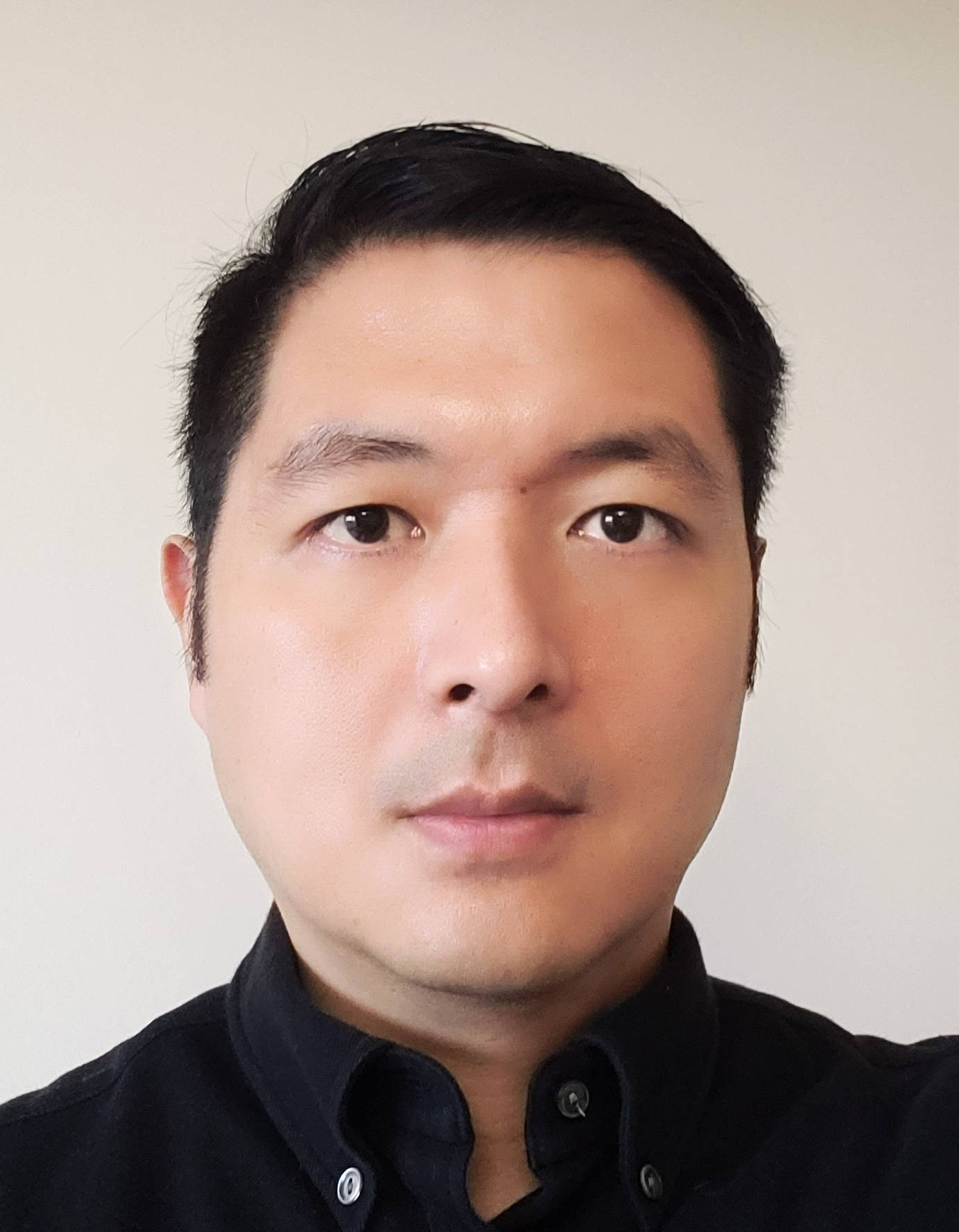}}]{Gene Cheung} (M'00--SM'07--F'21)
received the B.S. degree in electrical engineering from Cornell University in 1995, and the M.S. and Ph.D. degrees in electrical engineering and computer science from the University of California, Berkeley, in 1998 and 2000, respectively. 
He was a senior researcher in Hewlett-Packard Laboratories Japan, Tokyo, from 2000 till 2009. 
He was an assistant then associate professor in National Institute of Informatics (NII) in Tokyo, Japan, from 2009 till 2018. 
He is now a professor in York University, Toronto, Canada.
His research interests include 3D imaging and graph signal processing. 
He is a co-author of several paper awards and nominations, including the best student paper finalist in ICASSP 2021, best student paper award in ICIP 2013, ICIP 2017 and IVMSP 2016, best paper runner-up award in ICME 2012, and IEEE Signal Processing Society (SPS) Japan best paper award 2016. 
He is a recipient of the Canadian NSERC Discovery Accelerator Supplement (DAS) 2019. He is a fellow of IEEE.
\end{IEEEbiography}

\begin{IEEEbiography}[{\includegraphics[width=1in,height=1.25in,clip,keepaspectratio]{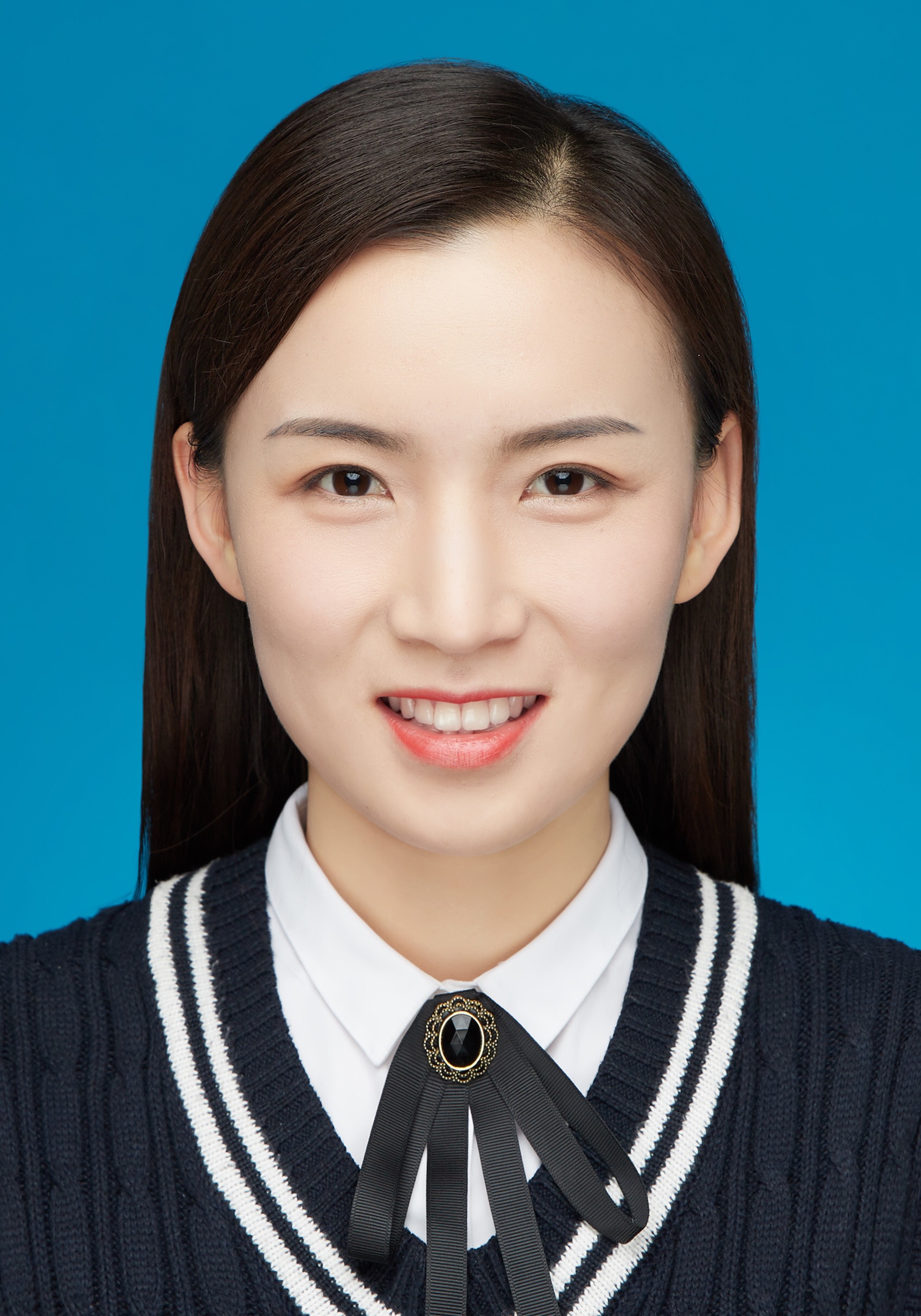}}]{Xue Zhang} (Member, IEEE) received the Ph.D. degree in signal and information processing from Beijing Jiaotong Unversity (BJTU), Beijing, China, in 2019. From 2015 to 2017, she was a Visiting Ph.D. student with the Signal Processing Laboratory (LTS4), Swiss Federal Institute of Technology (EPFL), Lausanne, Switzerland. In 2018, she was a Visiting Ph.D. student with National Institute of Informatics (NII), Tokyo, Japan. She was a Post-Doctoral Fellow with York University, Toronto, Canada, from 2019 till 2022. She is now a tenure-track Associate Professor in Shandong University of Science and Technology (SDUST), Qingdao, China. 
Her research interests include 3D image/video processing, interactive media navigation, and graph signal processing.
\end{IEEEbiography}